\documentclass{llncs}
\usepackage[bookmarks]{hyperref}
\usepackage{amsmath,amsfonts,amssymb,comment}

\usepackage{textcomp}
\usepackage{nicefrac}
\usepackage{stmaryrd}
\usepackage{longtable}
\usepackage{cases}
\usepackage{algorithmicx}
\usepackage[noend]{algpseudocode}
\usepackage{times}
\usepackage{paralist}

\usepackage{algorithm}
\usepackage{algpseudocode}
\usepackage{proof}

\makeatletter
\def\moverlay{\mathpalette\mov@rlay}
\def\mov@rlay#1#2{\leavevmode\vtop{%
   \baselineskip\z@skip \lineskiplimit-\maxdimen
   \ialign{\hfil$\m@th#1##$\hfil\cr#2\crcr}}}
\newcommand{\charfusion}[3][\mathord]{
    #1{\ifx#1\mathop\vphantom{#2}\fi
        \mathpalette\mov@rlay{#2\cr#3}
      }
    \ifx#1\mathop\expandafter\displaylimits\fi}
\makeatother

\makeatletter
\newsavebox{\@brx}
\newcommand{\llangle}[1][]{\savebox{\@brx}{\(\m@th{#1\langle}\)}%
  \mathopen{\copy\@brx\kern-0.5\wd\@brx\usebox{\@brx}}}
\newcommand{\rrangle}[1][]{\savebox{\@brx}{\(\m@th{#1\rangle}\)}%
  \mathclose{\copy\@brx\kern-0.5\wd\@brx\usebox{\@brx}}}
\makeatother

\newcommand{\game}{G}
\newcommand{\states}{S}

\newcommand{\act}{A}
\newcommand{\trans}{\delta}
\newcommand{\ap}{\mathsf{AP}}
\newcommand{\straa}{\sigma}
\newcommand{\straas}{\Sigma}
\newcommand{\strab}{\theta}
\newcommand{\strabs}{\Theta}

\newcommand{\almost}{\mathsf{\langle Almost \rangle}}
\newcommand{\positive}{\mathsf{\langle Positive \rangle}}
\newcommand{\boolc}{\mathsf{BoolC}}

\newcommand{\oneatl}{\operatorname{1-ATL}}
\newcommand{\onetwoatl}{\operatorname{(1,2)-ATL}}
\newcommand{\ctl}{\operatorname{CTL}}
\newcommand{\catl}{\operatorname{C-ATL}}
\newcommand{\atl}{\operatorname{ATL}}
\newcommand{\qatl}{\operatorname{QCTL}}
\newcommand{\pctl}{\operatorname{pCTL}}
\newcommand{\turn}{\mathsf{turn}}

\newcommand{\wt}{\widetilde}
\newcommand{\M}{\mathcal{M}}
\newcommand{\abs}{\mathsf{abs}}
\newcommand{\con}{\mathsf{con}}
\newcommand{\Cex}{\mathsf{Cex}}

\newcommand{\lab}{\mathcal{L}}
\newcommand{\pat}{\omega}
\newcommand{\Paths}{\Omega}
\newcommand{\PQ}{\mathrm{PQ}}
\newcommand{\calf}{\mathcal{F}}

\newcommand{\distr}{\mathcal{D}}

\newcommand{\av}{\mathsf{Av}}
\newcommand{\set}[1]{\{#1\}}
\newcommand{\supp}{\mathrm{Supp}}
\newcommand{\Next}{\varbigcirc}
\newcommand{\until}{\, \mathcal{U}}
\newcommand{\wrel}{\mathcal{W}}

\newcommand{\prb}{\mathrm{Pr}}
\newcommand{\true}{\mathsf{true}}
\newcommand{\false}{\mathsf{false}}
\newcommand{\plays}{\mathsf{Plays}}
\newcommand{\apre}{{\textsf{Apre}}}
\newcommand{\wb}{\overline}
\newcommand{\wh}{\widehat}

\newcommand{\nat}{\mathbb{N}}
\newcommand{\alt}{\mathcal{A}}
\newcommand{\qual}{\mathcal{C}}
\newcommand{\txtalt}{\mathsf{Alt}}
\newcommand{\txtsim}{\mathsf{Sim}}
\newcommand{\simul}{\mathcal{S}}
\newcommand{\dual}{\mathcal{M}}

\newcommand{\qualsim}[2]{#1 \preccurlyeq_{\qual} #2}		
\newcommand{\simgame}{\sim_\simul}
\newcommand{\altgame}{\sim_\alt}
\newcommand{\qualgame}{\sim_\qual}
\newcommand{\dualgame}{\sim_\dual}
\newcommand{\altabs}[2]{Abs^#2_{\mathcal{A}}(#1)}				
\newcommand{\simabs}[2]{Abs^#2_{\mathcal{S}}(#1)}				
\newcommand{\tsucc}{\textsf{Succ}}					
\newcommand{\conc}{\textsf{Conc}}						
\newcommand{\Part}{\Pi}							
\newcommand{\ppart}{\pi}							
\newcommand{\altcomp}{\textsf{A}}							
\newcommand{\simcomp}{\textsf{S}}							
\newcommand{\ov}{\overline}
\newcommand{\AL}{\mathsf{Alt}}
\newcommand{\SI}{\mathsf{Sim}}
\newcommand{\maxqual}{\qual_{\max}}

\usepackage{tikz}
\usetikzlibrary{fit}
\usetikzlibrary{backgrounds}
\usetikzlibrary{automata}
\usetikzlibrary{shapes}
\usetikzlibrary{matrix}
\usetikzlibrary{fit}
\usetikzlibrary{calc}
\usetikzlibrary{positioning}

\title{CEGAR for Qualitative Analysis of \\ Probabilistic 
Systems\thanks{The research was partly supported by Austrian Science Fund (FWF) 
Grant No P 23499- N23, FWF NFN Grant No S11407-N23 and  S11403-N23 (RiSE), 
ERC Start grant (279307: Graph Games), Microsoft faculty fellows award, 
the ERC Advanced Grant QUAREM (Quantitative Reactive Modeling).}
}
\author{Krishnendu Chatterjee \and Martin Chmel\'ik \and Przemys\l aw Daca}
\institute{IST Austria}
\begin{document}

\maketitle

\begin{abstract}
We consider Markov decision processes (MDPs) which are a standard model
for probabilistic systems. We focus on qualitative properties for MDPs 
that can express that desired behaviors of the system arise almost-surely 
(with probability~1) or with positive probability.
We introduce a new simulation relation to capture the refinement relation
of MDPs with respect to qualitative properties, and present discrete graph 
theoretic algorithms with quadratic complexity to compute the simulation 
relation.
We present an automated technique for assume-guarantee style reasoning 
for compositional analysis of MDPs with qualitative properties by 
giving a counterexample guided abstraction-refinement approach to compute
our new simulation relation. 
We have implemented our algorithms and show that the compositional analysis
leads to significant improvements. 
\end{abstract}

\section{Introduction}
\noindent{\bf Markov decision processes.}
\emph{Markov decision processes (MDPs)} are standard models for
analysis of probabilistic systems that exhibit both probabilistic 
and non-deterministic behavior~\cite{Howard,FV97}.
In verification of probabilistic systems, MDPs have been adopted as models 
for concurrent probabilistic systems~\cite{CY95},  probabilistic systems 
operating in open environments~\cite{SegalaT}, under-specified probabilistic 
systems~\cite{BdA95}, and applied in diverse domains~\cite{BaierBook,prism}
such as analysis of randomized communication and security protocols, stochastic 
distributed systems, biological systems, etc.

\smallskip\noindent{\bf Compositional analysis and CEGAR.}
One of the key challenges in analysis of probabilistic systems (as in 
the case of non-probabilistic systems) is the \emph{state explosion}
problem \cite{ClarkeBook}, as the size of concurrent systems grows exponentially in the
number of components. 
One key technique to combat the state explosion problem is the 
\emph{assume-guarantee} style composition reasoning~\cite{Pnu85a}, where 
the analysis problem is decomposed into components and the results
for components are used to reason about the whole system, instead 
of verifying the whole system directly.
For a system with two components, the compositional reasoning can be captured as the
following simple rule: consider a system with two components $G_1$ and $G_2$, 
and a specification $G'$ to be satisfied by the system; 
if $A$ is an abstraction of $G_2$ (i.e., $G_2$ refines $A$) and 
$G_1$ in composition with $A$ satisfies $G'$, then the composite systems of
$G_1$ and $G_2$ also satisfies $G'$. 
Intuitively, $A$ is an assumption on $G_1$'s environment that can be ensured by $G_2$.
This simple, yet elegant asymmetric rule is very effective in practice, specially
with a \emph{counterexample guided abstraction-refinement} (CEGAR) loop~\cite{Clarke00}.
There are many symmetric~\cite{PGBCB08} as well as circular 
compositional reasoning~\cite{AdHJ01,PGBCB08,KNPQ10} rules; however the simple asymmetric rule is most 
effective in practice and extensively studied, mostly for  
non-probabilistic systems~\cite{PGBCB08,FengKP11,CCST05,HJMQ03}.

\smallskip\noindent{\bf Compositional analysis for probabilistic systems.}
There are many works that have studied the abstraction-refinement and 
compositional analysis for probabilistic systems~\cite{CV10,HWZ08,KNP06,EKVY08}. 
Our work is most closely related to and inspired by~\cite{Komuravelli12} where a CEGAR approach was presented 
for analysis of MDPs (or labeled probabilistic transition systems); and 
the refinement relation was captured by \emph{strong simulation} that captures
the logical relation induced by safe-pCTL~\cite{HJ94,BerkP95,BdA95}.

\smallskip\noindent{\bf Qualitative analysis and its importance.}
In this work we consider the fragment of pCTL$^*$~\cite{HJ94,BerkP95,BdA95} that is relevant for 
\emph{qualitative analysis}, and refer to this fragment as $\qatl^*$. 
The qualitative analysis for probabilistic systems refers to \emph{almost-sure}
(resp. \emph{positive}) properties that are satisfied with probability~1 
(resp. positive probability).
The qualitative analysis for probabilistic systems is an important problem in 
verification that is of interest independent of the quantitative analysis problem.
There are many applications where we need to know whether the correct 
behavior arises with probability~1.
For instance, when analyzing a randomized embedded scheduler, we are
interested in whether every thread progresses with probability~1~\cite{CdAFMR13}.
Even in settings where it suffices to satisfy certain specifications with 
probability $\lambda<1$, the correct choice of $\lambda$ is a challenging 
problem, due to the simplifications introduced during modeling.
For example, in the analysis of randomized distributed algorithms it is 
quite common to require correctness with probability~1 
(see, e.g.,~\cite{PSL00,Sto02b}). 
Furthermore, in contrast to quantitative analysis, qualitative analysis is robust to 
numerical perturbations and modeling errors in the transition probabilities.
The qualitative analysis problem has been extensively studied for many probabilistic 
models, such as for MDPs~\cite{CH11,CH12,CHJS11}, perfect-information stochastic 
games~\cite{CJH03,ChaThesis}, concurrent stochastic games~\cite{dAHK98,CdAH11}, 
partial-observation MDPs~\cite{BBG08,CT12,CCT13,CDH10}, and partial-observation stochastic 
games~\cite{CDHR06,BGG09,CD12,CDH13,NV13,CDNV14}.

\smallskip\noindent{\bf Our contributions.} 
In this work we focus on the compositional reasoning of probabilistic systems 
with respect to qualitative properties, and our main contribution is a 
CEGAR approach for qualitative analysis of probabilistic systems. 
The details of our contributions are as follows:

\begin{enumerate}
\item To establish the logical relation induced by $\qatl^*$ we consider  
the logic $\atl^*$ for two-player games and the two-player game interpretation 
of an MDP where the probabilistic choices are 
resolved by an adversary.
In case of non-probabilistic systems and games there are two classical notions
for refinement, namely, \emph{simulation}~\cite{Milner71} and \emph{alternating-simulation}~\cite{AHKV98}.
We first show that the logical relation induced by $\qatl^*$ is \emph{finer} 
than the intersection of simulation and alternating simulation.
We then introduce a new notion of simulation, namely, \emph{combined simulation}, and show
that it captures the logical relation induced by $\qatl^*$.

\item We show that our new notion of simulation, which captures the logic relation of
$\qatl^*$, can be computed using discrete graph theoretic algorithms in quadratic 
time. 
In contrast, the current best known algorithm for strong simulation is polynomial 
of degree seven and requires numerical algorithms.
The other advantage of our approach is that it can be applied uniformly both to 
qualitative analysis of probabilistic systems as well as analysis of two-player
games (that are standard models for open non-probabilistic systems).

\item We present a CEGAR approach for the computation of combined simulation, and 
the counterexample analysis and abstraction refinement is achieved using the ideas
of~\cite{Henzinger03} proposed for abstraction-refinement for games.

\item We have implemented our approach both for qualitative analysis of MDPs as well
as games, and experimented on a number of well-known examples of MDPs and games.
Our experimental results show that our method achieves significantly better performance
as compared to the non-compositional verification as well as compositional analysis 
of MDPs with strong simulation. 

\end{enumerate}

\smallskip\noindent{\bf Related works.}
Compositional and assume-guarantee style reasoning has been extensively studied mostly in the 
context of non-probabilistic systems~\cite{PGBCB08,FengKP11,CCST05,HJMQ03}.
Game-based abstraction refinement has been studied in the context of probabilistic
systems~\cite{KNP06}.
The CEGAR approach has been adapted to probabilistic systems for reachability~\cite{HWZ08}
and safe-pCTL~\cite{CV10} under monolithic (non-compositional) abstraction refinement.
The work of~\cite{Komuravelli12} considers CEGAR for compositional analysis of probabilistic 
system with strong simulation.
An abstraction-refinement algorithm for a class of quantitative properties was 
studied in~\cite{DArgenioJJL01,DArgenioJJL02} and also 
implemented~\cite{raptureTool}.
Our logical characterization of the simulation relation is similar in spirit to \cite{Cleaveland91}, which shows how a fragment of the modal $\mu$-calculus can be used to efficiently decide behavioral preorders between components.
Our work focuses on CEGAR for compositional analysis of probabilistic systems 
for qualitative analysis: we characterize the required simulation relation; 
present a CEGAR approach for the computation of the simulation relation; 
and show the effectiveness of our approach both for qualitative analysis of MDPs 
and games.

\smallskip\noindent{\bf Organization of the paper.}
In Section~\ref{sec:defn} we present the basic definitions of games and logic for games.
In Section~\ref{sec:alt_sim} we introduce a new simulation relation for games, show that
it is finer than both simulation and alternating simulation, and present algorithms
to compute the relation.
In Section~\ref{sec:mdp} we present the definitions of MDPs and qualitative logics, and in Section~\ref{sec:mdplogic}
show that the logical relation induced by the qualitative logics on MDPs can be obtained
through our simulation relation introduced in Section~\ref{sec:alt_sim}.
In Section~\ref{sec:cegar} we present a CEGAR approach for our simulation relation 
and present experimental results in Section~\ref{sec:impl}.



\section{Game Graphs and Alternating-time Temporal Logics}
\label{sec:defn}
\noindent{\bf Notations.}
Let $\ap$ denote a non-empty finite set of atomic propositions. 
Given a finite set $\states$ we will denote
by $\states^{*}$ (respectively $\states^{\omega}$) the set of finite (resp. infinite) 
sequences of elements from $\states$, and let 
$\states^+=\states^* \setminus \set{\epsilon}$, where $\epsilon$ is the empty string.

\subsection{Two-player Games}
\smallskip\noindent\textbf{Two-player games.}
A \emph{two-player} game is a tuple $\game = (\states,\act,\av,\trans,\lab,s_0)$, where
\begin{itemize}
\item $\states$ is a finite set of states.
\item $\act$ is a finite set of actions.
\item $\av: \states \rightarrow 2^{\act} \setminus \emptyset$ is an \emph{action-available} function that assigns to every state $s \in \states$ 
the set $\av(s)$ of actions available in $s$.
\item $\trans: \states \times \act\rightarrow 2^\states \setminus \emptyset$ is a non-deterministic \emph{transition} function that given a state $s \in \states$ and an action
$a \in \av(s)$ gives the set $\trans(s,a)$ of successors of $s$ given action $a$.
\item $\lab: \states \rightarrow 2^{\ap}$ is a \emph{labeling} function that labels the states $s \in \states$ with the set $\lab(s)$ of atomic propositions true at $s$. 
\item $s_0 \in \states$ is an initial state.
\end{itemize}

\smallskip\noindent\textbf{Alternating games.}
A two-player game $\game$ is \emph{alternating} if in every state either 
Player~1 or Player~2 can make choices.
Formally, for all $s\in S$ we have 
either (i)~$|\av(s)|=1$ (then we refer to $s$ as a Player-2 state);
or (ii)~for all $a \in \av(s)$ we have $|\trans(s,a)|=1$ (then we refer
to $s$ as a Player-1 state).
For technical convenience we consider that in the case of alternating games, 
there is an atomic proposition $\turn \in \ap$ such that for every Player-1 state $s$ 
we have $\turn \in \lab(s)$, and for every Player~2 state $s'$ we have $\turn \not \in \lab(s')$.


\smallskip\noindent\textbf{Plays.}
A two-player game is played for infinitely many rounds as follows:
the game starts at the initial state, and in every round Player~1 chooses
an available action from the current state and then Player~2 chooses a
successor state, and the game proceeds to the successor state for the next round.
Formally, a \emph{play} in a two-player game is an infinite sequence 
$\pat = s_0 a_0 s_1 a_1 s_2 a_2 \cdots$ of states and actions such that for all 
$i\geq 0$ we have that $a_i \in \av(s_i)$  and $s_{i+1} \in \trans(s_i,a_i)$. 
We denote by $\Paths$ the set of all plays.

\smallskip\noindent{\bf Strategies.}
Strategies are recipes that describe how to extend finite prefixes of plays.
Formally, a \emph{strategy} for Player~1 is a function $\straa: (\states \times \act)^* \times S \rightarrow \act$, 
that given a finite history $w \cdot s \in (\states \times \act)^* \times S$ of the game  
gives an action from $\av(s)$ to be played next. 
We write $\straas$ for the set of all Player-$1$ strategies.
A strategy for Player~2 is a function $\strab: (\states \times \act)^+ \rightarrow \states$, 
that given a finite history $w \cdot s \cdot a$  of a play selects a successor state from the set 
$\trans(s,a)$. 
We write $\strabs$ for the set of all Player-$2$ strategies. 
\emph{Memoryless} strategies are independent of the history, but depend only on the
current state for Player~1 (resp.\ the current state and action for Player~2) and hence
can be represented as functions $\states \rightarrow \act$ for Player~1 
(resp.\ as functions  $\states \times \act \rightarrow \states$ for Player~2).

\smallskip\noindent{\bf Outcomes.}
Given a strategy $\straa$ for Player~1 and $\strab$ for Player~2 the \emph{outcome} is a unique play,
denoted as $\plays(s,\straa,\strab)= s_0 a_0 s_1 a_1 \cdots$, which is defined as follows: 
(i)~$s_0=s$; and (ii)~for all $i \geq 0$ we have $a_i =\straa(s_0 a_0 \ldots s_i)$ and $s_{i+1} = \strab(s_0 a_0 \ldots s_i a_i)$.
Given a state $s \in \states$ we denote by $\plays(s,\straa)$ (resp. $\plays(s,\strab))$ 
the set of possible plays given $\straa$ (resp. $\strab$), i.e., $\bigcup_{\strab' \in \strabs} \plays(s,\straa,\strab')$
(resp. $\bigcup_{\straa' \in \straas} \plays(s,\straa',\strab)$).

\smallskip\noindent\textbf{Parallel composition of two-player games.}
Given games $\game = (\states,\act,\av,\trans,\lab, s_0)$ and 
$\game' = (\states',\act, \av', \trans',\lab',s'_0)$  the \emph{parallel composition} of the games 
$\game\parallel\game' = (\wb{\states},\act,\wb{\av},\wb{\trans},\wb{\lab},\wb{s}_0)$ is defined as follows:

\begin{itemize}
\item The states of the composition are $\wb{\states} = \states \times \states'$.
\item The set of actions does not change with the composition.
\item For all $(s,s')$ we have $\wb{\av}((s,s'))=\av(s) \cap \av'(s')$.
\item The transition function for a state $(s,s') \in \wb{\states}$ and an action $a \in \wb{\av}((s,s'))$ is defined as $\wb{\trans}((s,s'),a) =
\{ (t,t') \mid t \in \trans(s,a) \wedge t' \in \trans'(s',a)\}$.
 \item The labeling function $\wb{\lab}((s,s'))$ is defined as $\lab(s)\cup \lab'(s')$.
 \item The initial state is $\wb{s}_0=(s_0,s'_0).$
\end{itemize}

\begin{remark}
For simplicity we assume that the set of actions in both components is identical, and for every pair of states the
 intersection of their available actions is non-empty.
Parallel composition can be extended to cases where the sets of actions are 
different~\cite{RajeevTomBook}.
\end{remark}

\subsection{Alternating-time Temporal Logic}
We consider the Alternating-time Temporal Logic ($\atl^*$)~\cite{AHK02} as a logic to specify 
properties for two-player games.

\smallskip\noindent\textbf{Syntax.} 
The syntax of the logic is given in positive normal form by defining the set of \emph{path formulas}~$(\varphi)$ 
and \emph{state formulas}~$(\psi)$ according to the following grammar:
\begin{eqnarray*}
\text{state formulas:} & \qquad \psi & ::=   q \mid \neg q \mid \psi \vee \psi \mid \psi \wedge \psi \mid \PQ(\varphi)  \\
\text{path formulas:} & \qquad \varphi & ::=   \psi \mid \varphi \vee \varphi \mid \varphi \wedge \varphi \mid \Next \varphi \mid \varphi \until \varphi \mid \varphi \wrel \varphi;
\end{eqnarray*}
where $q \in \ap$ is an atomic proposition and $\PQ$ is a path quantifier. 
The operators $\Next$ (next), $\until$ (until), and $\wrel$ (weak until) are the temporal operators. We will use $\true$ as a shorthand for $q \vee \neg q$ and $\false$ for $q \wedge \neg q$ for some $q \in \ap$.
The path quantifiers $\PQ$ are as follows:
 $$\atl^* \text{ path quantifiers: } \llangle 1 \rrangle, \llangle 2 \rrangle,\llangle 1,2 \rrangle \text{, and } \llangle \emptyset \rrangle. $$

\smallskip\noindent\textbf{Semantics.}
Given a play $\pat = s_0 a_0 s_1 a_1 \cdots$ we denote by $\pat[i]$ the suffix starting at the $i$-th state element of the play 
$\pat$, i.e., $\pat[i] = s_i a_i s_{i+1} a_{i+1} \cdots$. The semantics of path formulas is defined inductively as follows:

\begin{center}
\begin{tabular}{ll}
$\pat \models \psi $ & $\qquad \text{iff } \pat[0] \models \psi $\\
$\pat \models \varphi_1 \vee \varphi_2 $& $\qquad \text{iff } \pat \models \varphi_1 \text{ or } \pat \models \varphi_2$ \\
$\pat \models \varphi_1 \wedge \varphi_2 $& $\qquad \text{iff } \pat \models \varphi_1 \text{ and } \pat \models \varphi_2$ \\
$\pat \models  \Next \varphi $& $\qquad  \text{iff } \pat[1]  \models \varphi$ \\
$\pat \models  \varphi_1 \until \varphi_2 $& $\qquad  \text{iff } \exists j \in \nat: \pat[j]  \models \varphi_2 \text{ and } \forall 0 \leq i < j: \pat[i] \models \varphi_1$\\
$\pat \models \varphi_1 \wrel \varphi_2	$& $\qquad \text{iff } \varphi_1 \until \varphi_2 \text{ or }\forall j \in \nat: \pat[j] \models \varphi_{1}$.
\end{tabular}
\end{center}
Given a path formula $\varphi$, we denote by $\llbracket \varphi \rrbracket_\game$ the set of plays $\pat$ such that $\pat \models \varphi$. We omit the $\game$ lower script when the game is clear from context. 
The semantics of state formulas for $\atl^*$ is defined 
as follows:
\begin{center}
\begin{tabular}{ll}
$s \models q $ & $\qquad \text{iff } q \in \lab(s)$\\
$s \models \neg q $ & $\qquad \text{iff } q \not \in \lab(s)$\\
$s \models \psi_1 \vee \psi_2 $& $\qquad \text{iff } s \models \psi_1 \text{ or } s \models \psi_2$ \\
$s \models \psi_1 \wedge \psi_2 $& $\qquad \text{iff } s \models \psi_1 \text{ and } s \models \psi_2$ \\
$s \models \llangle 1 \rrangle (\varphi) $ & $\qquad \text{iff } \exists \straa \in \straas, \forall \strab \in \strabs: \plays(s, \straa,\strab) \in \llbracket \varphi \rrbracket$\\
$s \models \llangle 2 \rrangle (\varphi) $ & $\qquad \text{iff } \exists \strab \in \strabs, \forall \straa \in \straas: \plays(s, \straa,\strab) \in \llbracket \varphi \rrbracket$\\
$s \models \llangle 1,2 \rrangle (\varphi) $ & $\qquad \text{iff } \exists \straa \in \straas, \exists \strab \in \strabs: \plays(s, \straa,\strab) \in \llbracket \varphi \rrbracket$\\
$s \models \llangle \emptyset \rrangle (\varphi) $ & $\qquad \text{iff } \forall \straa \in \straas, \forall \strab \in \strabs: \plays(s, \straa,\strab) \in \llbracket \varphi \rrbracket$;
\end{tabular}
\end{center}
where $s \in S$ and $q \in \ap$.
Given an $\atl^*$ state formula $\psi$ and a two-player game $\game$, 
we denote by $\llbracket \psi \rrbracket_\game = \{ s \in \states \mid s \models \psi \}$ 
the set of states that satisfy the formula $\psi$. We omit the $\game$ lower script when the game is clear from context.

\smallskip\noindent\textbf{Logic fragments.} We define several fragments of the logic $\atl^*$:

\begin{itemize}
\item \emph{Restricted temporal operator use.}
An important fragment of $\atl^*$ is $\atl$ where every temporal operator is immediately preceded by a path quantifier.

\item \emph{Restricting path quantifiers.}
We also consider fragments of $\atl^*$ (resp. $\atl$) where the path quantifiers are restricted.
We consider (i)~$1$-fragment (denoted $\oneatl^*$) where only $\llangle 1 \rrangle$ path quantifier is used;
(ii)~the $(1,2)$-fragment (denoted $\onetwoatl^*$) where only $\llangle 1,2 \rrangle$ path quantifier is used;
and (iii)~the combined fragment (denoted $\catl^*$) where both $\llangle 1 \rrangle$ and $\llangle 1,2 \rrangle$ 
path quantifiers are used.
We use a similar notation for the respective fragments of $\atl$ formulas.
\end{itemize}

\smallskip\noindent{\bf Logical characterization of states.}
Given two games  $\game$ and $\game'$, 
and a logic fragment $\calf$ of $\atl^*$, we consider the following relations on the state space induced by the 
logic fragment $\calf$:
$$ \preccurlyeq_{\calf}(\game,\game') = \{(s,s') \in \states \times \states' \mid \forall \psi \in \calf: \text{ if } s \models \psi \text{ then } s' \models \psi\}; $$
and when the games are clear from context we simply write $\preccurlyeq_{\calf}$ for $\preccurlyeq_{\calf}(\game,\game')$.
We will use the following notations for the relation induced by the logic fragments we consider:
(i)~$\preccurlyeq_1^*$ (resp. $\preccurlyeq_1$) for the relation induced by the $\oneatl^*$ (resp. $\oneatl)$ fragment;
(ii)~$\preccurlyeq_{1,2}^*$ (resp. $\preccurlyeq_{1,2}$) for the relation induced by the $\onetwoatl^*$ (resp. $\onetwoatl)$ fragment; and
(iii)~$\preccurlyeq_{C}^*$ (resp. $\preccurlyeq_{C}$) for the relation induced by the $\catl^*$ (resp. $\catl)$ fragment.
Given $\game$ and $\game'$ we can also consider $\game''$ which is the disjoint union of the two games,
and consider the relations on $\game''$; and hence we will often consider a single game as input for the 
relations.

\section{Combined Simulation Relation Computation}
\label{sec:alt_sim}
In this section we first recall the notion of simulation~\cite{Milner71}
and alternating simulation~\cite{AHKV98}; and then present a new notion of 
\emph{combined simulation}.

\smallskip\noindent\textbf{Simulation.}
Given two-player games 
$\game = (\states,\act,\av,\trans,\lab,s_0)$ and 
$\game' = (\states',\act',\av',\trans',\lab',s'_0)$,
a relation $\simul \subseteq \states \times \states'$ is a \emph{simulation} from $\game$ to $\game'$ 
if for all $(s, s') \in \simul$  the following conditions hold:
\begin{enumerate}
\item \emph{Proposition match:} The atomic propositions match, i.e., $\lab(s) = \lab'(s')$.
\item \emph{Step-wise simulation condition:} For all actions $a \in \av(s)$ and states $t \in \trans(s,a)$ there exists an action $a' \in \av'(s')$ and a state $t' \in \trans(s',a')$ such that $(t,t') \in \simul$.
\end{enumerate}
We denote by $\simul_{\max}^{\game,\game'}$ the largest simulation relation between the two games
(we write $\simul_{\max}$ instead of  $\simul_{\max}^{\game,\game'}$ when $\game$ and $\game'$ are clear from the context).
We write $\game \simgame \game'$ when $(s_0, s'_0)\in \simul_{\max}$.
The largest simulation relation characterizes the logic relation of $\onetwoatl$ and $\onetwoatl^*$:
the $\onetwoatl$-fragment interprets a game as a transition system and the formulas coincide with 
existential $\ctl$, and hence the logic characterization follows from the classical results on 
simulation and $\ctl$~\cite{Milner71,RajeevTomBook}.

\begin{proposition}
For all games $\game$ and $\game'$ we have $\simul_{\max}=\preccurlyeq_{1,2}^*=\preccurlyeq_{1,2}$.
\end{proposition}


\smallskip\noindent\textbf{Alternating simulation.}
Given two games
$\game = (\states,\act,\av,\trans,\lab,s_0)$ and 
$\game' = (\states',\act',\av',\trans',\lab',s'_0)$,
a relation $\alt \subseteq \states \times \states'$ is an \emph{alternating simulation} from $G$ to $G'$ if for all $(s, s') 
\in \alt$ the following conditions hold:
\begin{enumerate}
\item  \emph{Proposition match:} The atomic propositions match, i.e., $\lab(s) = \lab'(s')$.
\item \emph{Step-wise alternating-simulation condition:} For all actions $a \in \av(s)$ there exists an action $a' \in \av'(s')$ such that for all states $t' \in \trans'(s',a')$ there exists a state $t \in \trans(s,a)$ such that $(t,t') \in \alt$.
 \end{enumerate}
We denote by $\alt_{\max}^{\game,\game'}$ the largest alternating-simulation relation between the two games
(we write $\alt_{\max}$ instead of  $\alt_{\max}^{\game,\game'}$ when $\game$ and $\game'$ are clear from the context). 
We write $\game \altgame \game'$ when $(s_0, s'_0)\in \alt_{\max}$.
The largest alternating-simulation relation characterizes the logic relation of $\oneatl$ and $\oneatl^*$~\cite{AHKV98}.

\begin{proposition}
For all games $\game$ and $\game'$ we have $\alt_{\max}=\preccurlyeq_{1}^*=\preccurlyeq_{1}$.
\end{proposition}

\smallskip\noindent\textbf{Combined simulation.}
We present a new notion of {combined simulation} 
that extends both simulation and alternating simulation, and we 
show how the combined simulation characterizes the logic relation 
induced by $\catl^*$ and $\catl$.
Intuitively, the requirements on the combined-simulation relation 
combine the requirements imposed by alternating simulation and simulation
in a step-wise fashion. 
Given two-player games $\game = (\states,\act,\av,\trans,\lab,s_0)$ and 
$\game' = (\states',\act',\av',\trans',\lab',s'_0)$, 
a relation $\qual \subseteq \states \times \states$ is a \emph{combined simulation} 
from $G$ to $G'$ if for all $(s, s') \in \qual$ the following conditions hold:
\begin{enumerate}
\item  \emph{Proposition match:} The atomic propositions match, i.e., $\lab(s) = \lab'(s')$.
\item  \emph{Step-wise simulation condition:} For all actions $a \in \av(s)$ and states $t \in \trans(s,a)$ there exists an action $a' \in \av'(s')$ and a state $t' \in \trans(s',a')$ such that $(t,t') \in \qual$.
\item  \emph{Step-wise alternating-simulation condition:} For all actions $a \in \av(s)$ there exists an action $a' \in \av'(s')$ such that for all states $t' \in \trans'(s',a')$ there exists a state $t \in \trans(s,a)$ such 
that $(t,t') \in \qual$.
\end{enumerate}
We denote by $\qual_{\max}^{\game,\game'}$ the largest combined-simulation relation between the two games
(and write $\qual_{\max}$ when $\game$ and $\game'$ are clear from the context).
We also write $\game \qualgame \game'$ when $(s_0, s'_0)\in \qual_{\max}$.
We first illustrate with an example that the logic relation $\preccurlyeq_{C}$
induced by $\catl$ is finer than the intersection of simulation and alternating-simulation 
relation; then present a game theoretic characterization of $\qual_{\max}$; and finally show that $\qual_{\max}$ gives the relations $\preccurlyeq_{C}^*$ and 
$\preccurlyeq_{C}$.

\begin{figure}[htb]
  \centering
  \scalebox{0.9}{\begin{tikzpicture}[auto, node distance=2cm,->, semithick,initial text=, align=left]
\tikzstyle{state}=[circle, draw]
\tikzstyle{initstate}=[state,initial]
\tikzstyle{transition}=[->,>=stealth']
\node [initstate, initial above] (s0) {$s_0$};
\node [state, fill=gray!40,right of=s0] (s1)  {$s_1$};
\node  (sl) at (0,1) {$\game$};

\node [state, , right of=s1] (t2) {$t_2$};
\node [initstate, initial above, right of=t2] (t0) {$t_0$};
\node [state, fill=gray!40, right of=t0] (t1) {$t_1$};
\node  (sr) at (6,1) {$\game'$};

\path
        (s0) edge[loop below] node[inner sep=0mm,pos=0.15] (s00) {}  node[above] {} (s0)
        (s0) edge  node[inner sep=0mm,pos=0.3] (s01) 	        {} node[above] {} (s1)
        (s0) edge[loop left] node[left]{$a_2$} (s0)
        (s1) edge[loop right] node[right]{$a_3$} (s1);
\path[-,shorten <=-1.5pt,shorten >=0mm] (s00) edge [bend right] node[right] {$a_1$} (s01) ;

\path
        (t0) edge[loop below] node[inner sep=0mm,pos=0.15] (t00) {}  node[above] {} (t0)
        (t0) edge  node[inner sep=0mm,pos=0.3] (t01) 	        {} node[above] {} (t1)
        (t0) edge[] node[above]{$a_2$} (t2)
        (t1) edge[loop right] node[right]{$a_3$} (t1)
        (t2) edge[loop below] node[right]{$a_2$} (t2);

\path[-,shorten <=-1.5pt,shorten >=0mm] (t00) edge [bend right] node[right] {$a_1$} (t01) ;

\end{tikzpicture} }
  \caption{Games $\game, \game'$ such that $\game \simgame \game'$ and $\game \altgame \game'$, but $\game \not\qualgame \game'$.}
  \label{fig:combined}
\end{figure}
\begin{example}
Consider the games $\game$ and $\game'$ shown in Figure~\ref{fig:combined}. 
White nodes are labeled by an atomic proposition $p$ and gray nodes by $q$.
The largest simulation and alternating-simulation relations between $\game$ and $\game'$ are:
$\simul_{\max}=\set{(s_0, t_0),(s_1, t_1)}, \alt_{\max}=\set{(s_0, t_0),(s_0, t_2),(s_1, t_1)}$. 
However, consider the formula $\psi = \llangle 1 \rrangle (\Next (p \land \llangle 1,2 \rrangle (\Next q)))$. We have
that $s_0 \models \psi$, but $t_0 \not \models \psi$. It follows that 
$(s_0,t_0) \not \in \preccurlyeq_{C}$.
\qed
\end{example}

\smallskip\noindent{\bf Combined-simulation games.} 
The simulation and the alternating-simulation relation can be obtained by solving 
two-player safety games~\cite{HHK95,AHKV98,CCK12}. 
We now define a two-player game for the combined-simulation relation characterization.
The game is played on the synchronized product of the two input games.
Given a state $(s,s')$, first Player~2 decides whether 
to check for the step-wise simulation condition or the step-wise  alternating-simulation 
condition. 
The step-wise simulation condition is checked by playing a two-step game, 
and the step-wise alternating-simulation condition is checked by playing a 
four-step game.
Consider two games $\game = (\states, \act, \av, \trans,\lab, s_0)$
and $\game' = ( \states', \act', \av', \trans',\lab',s'_0)$. 
We construct the \emph{combined-simulation game} 
$\game^{\qual} = (\states^\qual,\act^\qual,\av^\qual, \trans^\qual,\lab^\qual,s^\qual_0)$ as follows:

\begin{itemize}
\item \emph{The set of states.} The set of states $\states^\qual$ is:
\begin{eqnarray*}
\states^\qual = & &(\states \times \states') \cup (\states \times \states' \times \set{\SI} \times \set{1,2} )\cup (\states \times \states' \times \set{\AL} \times \set{2})\\
& \cup & (\states \times \states' \times \set{\AL} \times \act \times \set{1} )\cup 
(\states \times \states' \times \set{\AL} \times \act \times \act' \times \set{1,2} )
\end{eqnarray*}
Intuitively, in states in $S\times S'$ and in states where the last component is~2 
it is Player~2's turn to make the choice of successors, 
and in all other states Player~1 makes the choice of actions.

\item \emph{The set of actions.} The set of actions is as follows: $\act^\qual = \set{\bot} \cup \states \cup \states' \cup \act'$.

\item \emph{The transition function and the action-available function.} 
\begin{enumerate}

\item \emph{Choice of simulation or alternating-simulation.}
For a state $(s,s')$ we have only one action $\bot$ available for Player~1 and 
we have $\trans^{\qual}((s,s'),\bot)= \set{(s,s',\AL,2), (s,s',\SI,2)}$,
i.e., Player~2 decides whether to check for step-wise simulation or step-wise 
alternating-simulation conditions.

\item \emph{Checking step-wise simulation conditions.}
We describe the transitions for checking the simulation conditions:

\begin{enumerate}
\item For a state $(s,s',\SI,2)$ we have only one action $\bot$ available for Player~1 and 
we have $\trans^{\qual}((s,s',\SI,2),\bot)=\set{(t,s',\SI,1) \mid \exists a \in \av(s): \ t \in \trans(s,a)}$.

\item For a state $\ov{s}=(t,s',\SI,1)$ we have $\av^\qual(\ov{s})=\set{t'\mid \exists a'\in \av(s'): \  t'\in \trans'(s',a')}$
and $\trans^{\qual}(\ov{s},t')=\set{(t,t')}$.
\end{enumerate}
Intuitively, first Player~2 chooses an action $a \in \av(s)$ and a successor $t \in \trans(s,a)$ and 
challenges Player~1 to match, and Player~1 responds with an action $a'\in \av'(s')$ and 
a state $t'\in \trans'(s',a')$.

\item \emph{Checking  step-wise alternating-simulation conditions.}
We describe the transitions for checking the alternating-simulation conditions:
\begin{enumerate}
\item For a state $(s,s',\AL,2)$ we have only one action $\bot$ available for Player~1 and 
we have $\trans^{\qual}((s,s',\AL,2),\bot)=\set{(s,s',\AL,a,1) \mid  a \in \av(s)}$.
\item For a state $\ov{s}=(s,s',\AL,a,1)$ we have $\av^\qual(\ov{s})=\av'(s')$ and $\trans^{\qual}(\ov{s},a')=\set{(s,s',\AL,a,a',2)}$.
\item For a state $(s,s',\AL,a,a',2)$ we have only one action $\bot$ available for Player~1 and 
we have $\trans^{\qual}((s,s',\AL,a,a',2),\bot)=\set{(s,t',\AL,a,a',1) \mid  t' \in \trans'(s',a')}$.
\item For a state $\ov{s}=(s,t',\AL,a,a',1)$ we have $\av^\qual(\ov{s})= \trans(s,a)$ and $\trans^{\qual}(\ov{s},t)=\set{(t,t')}$.
\end{enumerate}
Intuitively, first Player~2 chooses an action $a$ from $\av(s)$ and Player~1 responds with an 
action $a'\in \av'(s')$ (in the first two-steps); 
then Player~2 chooses a successor $t'$ from $\trans'(s',a')$ and Player~1 responds by choosing 
a successor $t$ in $\trans(s,a)$.
\end{enumerate}

\item \emph{The labeling function.} The set of atomic proposition $\ap$ contains a single proposition $p \in \ap$. 
The labeling function $\lab^\qual$ given a state $\ov{s} \in \states^\qual$ is defined as follows:
$\lab^\qual(\ov{s})  = p$ iff  $\ov{s} = (s,s')$ and $\lab(s) \not = \lab'(s')$. 
Intuitively, Player~2's goal is to reach a state $(s,s')$ where 
the propositional labeling of the original games do not match,
i.e., to reach a state labeled $p$ by $\lab^\qual$.
\item \emph{The initial state.} The state $s^\qual_0$ is  $(s_0,s'_0)$.
\end{itemize}
In the combined simulation game we refer to Player~1 as the \emph{proponent} (trying to 
establish the combined simulation) and Player~2 as the \emph{adversary} (trying to violate the combined simulation).

\begin{figure}[htb]
  \centering
  \scalebox{0.9}{\begin{tikzpicture}[auto, node distance=2cm,->, semithick,initial text=, align=left]
\tikzstyle{state}=[ inner sep=0mm, font=\scriptsize]
\tikzstyle{pstate}=[ state, fill=gray!40]
\tikzstyle{slab}=[ inner sep=1mm, font=\scriptsize]
\tikzstyle{initstate}=[initial, state]
\tikzstyle{transition}=[->,>=stealth']

\def\vgap{-1cm}
\def\hgap{3cm}

\node [initstate, initial above] (ta) at (0,1*\vgap) 	{$(s_0, t_0)$};
\node [state] (tb) at (-1.625*\hgap,1*\vgap) 		{$(s_0, t_0, \AL, 2)$};
\node [state] (tc) at (1*\hgap,1*\vgap) 		{$(s_0, t_0, \SI, 2)$};
\node [state] (td) at (-2.25*\hgap,2*\vgap) 		{$(s_0, t_0, \AL, a_1, 1)$};
\node [state] (te) at (-1*\hgap,2*\vgap) 		{$(s_0, t_0, \AL, a_2, 1)$};
\node [slab] (tf) at (-2.25*\hgap,3*\vgap) 		{$\ldots$};
\node [state] (tg) at (-0.5*\hgap,3*\vgap) 		{$(s_0, t_0, \AL, a_2, a_1, 2)$};
\node [state] (th) at (-1.5*\hgap,3*\vgap) 		{$(s_0, t_0, \AL, a_2, a_2, 2)$};
\node [state] (ti) at (1.5*\hgap,3*\vgap) 		{$(s_0, t_0, \SI, 1)$};
\node [state] (tj) at (0.5*\hgap,3*\vgap) 		{$(s_1, t_0, \SI, 1)$};
\node [slab] (tk) at (-0.5*\hgap,4*\vgap) 		{$\ldots$};
\node [state] (tl) at (-1.5*\hgap,4*\vgap) 		{$(s_0, t_2, \AL, a_2, a_2, 1)$};
\node [slab] (tm) at (1.5*\hgap,4*\vgap) 		{$\dots$};
\node [state] (to) at (-1.5*\hgap,5*\vgap) 		{$(s_0, t_2)$};
\node [pstate] (tp) at (-0.5*\hgap,5*\vgap) 		{$(s_1, t_0)$};
\node [state] (tr) at (0.5*\hgap,5*\vgap) 		{$(s_1, t_1)$};
\node [pstate] (ts) at (1.5*\hgap,5*\vgap) 		{$(s_1, t_2)$};
\node [slab] (tt) at (-1.5*\hgap,6*\vgap) 		{$\ldots$};
\node [slab] (tu) at (-0.5*\hgap,6*\vgap) 		{$\ldots$};
\node [slab] (tw) at (0.5*\hgap,6*\vgap) 		{$\ldots$};
\node [slab] (tx) at (1.5*\hgap,6*\vgap) 		{$\ldots$};

\node [below=-1mm of ta] (la)  				{$\bot$};
\node [below=-1mm of tb] (lb)  				{$\bot$};
\node [below=0mm of tc] (lc)  				{$\bot$};

\path
	(ta) edge  (tb)
	(ta) edge  (tc)
	(tb) edge  (td)
	(tb) edge  (te)
        (tc) edge  (ti)
        (tc) edge  (tj)
        (te) edge node[right] {$a_1$} (tg)
        (te) edge node[left] {$a_2$} (th)
	(th) edge node[right] {$\bot$}(tl)
	(tl) edge  node[left] {$s_0$} (to)
        (tj) edge node[left] {$t_0$} (tp)
        (tj) edge node[left] {$t_1$} (tr)
        (tj) edge node[left] {$t_2$} (ts);

\path[dashed]
	(td) edge  (tf)
	(tg) edge  (tk)
	(ti) edge  (tm)
        (to) edge  (tt)
        (tp) edge  (tu)
        (tr) edge  (tw)
        (ts) edge  (tx);

\draw[-,thick]  (-0.6cm,\vgap) arc (180:360:0.6cm);
\draw[-,thick]  (-1.625*\hgap-0.6cm,\vgap-0.32cm) arc (215:325:0.7cm);
\draw[-,thick]  (\hgap-0.45cm,\vgap-0.6cm) arc (230:310:0.7cm);

\end{tikzpicture} }
  \caption{Part of the combined-simulation game of $\game$ and $\game'$ from Figure~\ref{fig:combined}.}
  \label{fig:csg}
\end{figure}
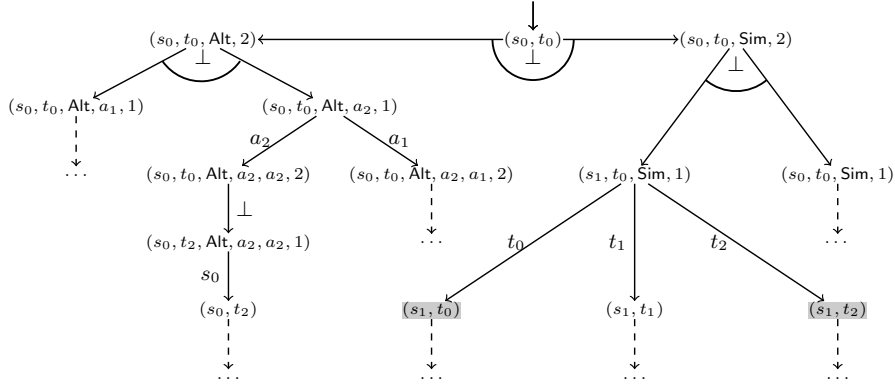

\begin{example}
A part of the combined-simulation game of  $\game$ and $\game'$ from Figure~\ref{fig:combined} is shown in Figure~\ref{fig:csg}.
Dashed arrows indicate that the successors of a given state are omitted in the figure.
Gray states are labeled by an atomic proposition $p$, hence are the goal states for the adversary.\qed
\end{example}

\smallskip\noindent{\bf Shorthand for safety objectives.}
We will use the following shorthand for \emph{safety} objectives:
$\Box \: \varphi \equiv \varphi \: \wrel \: \false$;
i.e., the formula $\Box \varphi$ is satisfied by paths where 
$\varphi$ is always true.

\begin{theorem}
For all games $\game$ and $\game'$ we have $\qual_{\max} = \llbracket \llangle 1 \rrangle (\Box \neg p) \rrbracket_{\game^\qual} \cap (\states \times \states')$.
\end{theorem}
\begin{proof}
The statement follows directly from the definition of combined simulation, and the fact that the game construction 
mimics the definition of combined simulation (as in the case of simulation and alternating simulation~\cite{HHK95,AHKV98,CCK12}). \qed
\end{proof}

\smallskip\noindent{\bf Winning strategies.}
Given a combined-simulation game $\game^\qual$ we say that a strategy $\straa$ for the proponent is 
\emph{winning} from a state $s$ if for all strategies $\strab$ of the adversary 
we have $\plays(s,\straa,\strab) \models \Box (\neg p)$.
A strategy $\strab$ for the adversary is \emph{winning} from state $s$ if for all strategies 
$\straa$ of the proponent we have $\plays(s,\straa,\strab) \models \true \until p$.
Whenever the proponent (resp. adversary) has a winning strategy, the proponent (resp. adversary) also has 
memoryless winning strategy~\cite{gradel2002automata}.

\smallskip\noindent{\bf Combined simulation logical characterization.}
Our next goal is to establish that combined simulation gives the logical characterization 
of $\catl^*$ and $\catl$.
To prove the result we first introduce the notion of equivalence between plays:
Given two plays $\pat=s_0 a_0 s_1 a_1 s_2 \cdots $ and $\pat'= s'_0 a'_0 s'_1 a'_1 s'_2 \cdots $ 
we write $\pat \sim_\qual \pat'$ if for all $i\geq0$ we have $(s_i,s'_i) \in \qual_{\max}$.

\begin{lemma}
\label{lem:transl_one}
Given two games $\game$ and $\game'$, let $\maxqual$ be the combined simulation.
For all $(s,s') \in \maxqual$ the following assertions hold: 
\begin{itemize}
\item For all Player~1 strategies $\straa$ in $\game$, 
there exists a Player~1  strategy $\straa'$ in $\game'$ such that for every play 
$\pat' \in \plays(s',\straa')$ there exists a play $\pat \in \plays(s,\straa)$ 
such that $\pat \sim_\qual \pat'$.
\item For all pair of strategies $\straa$ and $\strab$ in $\game$, 
there exists a pair of strategies $\straa'$ and $\strab'$ in $\game'$ 
such that $\plays(s,\straa,\strab) \sim_\qual \plays(s',\straa',\strab')$,
\end{itemize}
\end{lemma}
\begin{proof}
We present the details of the first item.
\begin{itemize}
\item 
Consider a winning strategy $\straa^{\qual}$ for the proponent in $\game^\qual$ such that for all $(s,s') 
\in \maxqual$ and against all strategies $\strab^\qual$ we have 
$\plays(s,\straa^\qual,\strab^\qual) \in \llbracket \Box(\neg p) \rrbracket$.
Given the Player~1 strategy $\straa$ in $\game$ we construct $\straa'$ in 
$\game'$ using the strategy $\straa^\qual$.
Consider a history $w \cdot s$ in $\game$ and $w'\cdot s' \in \game'$ 
such that $(s,s') \in \maxqual$.
Let $\straa(w\cdot s)=a$.
We define $\straa'(w'\cdot s')$ as follows.
Let $h$ be an arbitrary history in $\game^\qual$ that only visits state in 
$\maxqual$ and ends in $(s,s')$.
Let $a'=\straa^\qual(h \cdot (s,s',\AL,2) \cdot (s,s',\AL,a,2))$;
(i.e., the action played by the strategy $\straa^\qual$ in response 
to the choice of checking alternating simulation and the action $a$ by  
Player~2 in $\game^\qual$).
Then the strategy $\straa'$ plays accordingly,
i.e., $\straa'(w'\cdot s')=a'$.
In the next step for every choice $t'$ of the adversary there exists a choice $t$ of 
the proponent such that $\lab(t) = \lab'(t')$ and $(t,t') \in \maxqual$ and 
the matching can proceed.
\item The proof is similar to the first item, and instead of using the step-wise 
alternating-simulation gadget for strategy construction (of the first item) 
we use the step-wise simulation gadget from $\game^\qual$ to construct the 
strategy pairs.
\end{itemize}
The desired result follows.
\qed
\end{proof}

In the following theorem we establish the relation between combined simulation and the $\catl^*$ fragment of $\atl^*$. 
\begin{theorem}
\label{thm:equiv}
For all games $\game$ and $\game'$ we have
$\qual_{\max} = \preccurlyeq_{C}^*= \preccurlyeq_{C}$.
\end{theorem}

\begin{proof}
\noindent{\em First implication.}
We first prove the implication $\maxqual \subseteq \preccurlyeq_{C}^*$. 
We will show the following assertions:
\begin{itemize}
\item For all states $s$ and $s'$ such that $(s,s') \in \maxqual$, we have that every $\catl^*$ state formula satisfied in $s$ is also satisfied in $s'$.
\item For all plays $\pat$ and $\pat'$ such that $\pat \sim_{\qual} \pat'$, 
we have that every $\catl^*$ path formula satisfied in $\pat$ is also satisfied in $\pat'$.
\end{itemize}
We will prove the theorem by induction on the structure of the formulas.
The interesting cases for the induction step are formulas $\llangle 1 \rrangle (\varphi)$ and $\llangle 1,2 \rrangle (\varphi)$, 
where $\varphi$ is a path formula.
\begin{itemize}
\item Assume $s \models \llangle 1 \rrangle (\varphi)$ and $(s,s') \in \maxqual$. 
It follows that there exists a strategy $\straa \in \straas$ that ensures the path formula $\varphi$
from state $s$ against any strategy $\strab \in \strabs$. 
We want to show that $s' \models \llangle 1 \rrangle (\varphi)$. 
By Lemma~\ref{lem:transl_one}(item~1) we have that there exists a strategy $\straa'$ 
for Player~1 from $s'$ such that for every play $\pat' \in \plays(s',\straa')$ 
there exists a play $\pat \in \plays(s,\straa)$ such that $\pat \sim_{\qual} \pat'$. 
By inductive hypothesis we have that $s' \models \llangle 1 \rrangle (\varphi)$.

\item Assume $s \models \llangle 1,2 \rrangle (\varphi)$ and $\qual(s,s')$. 
It follows that there exist strategies $\straa \in \straas, \strab \in \strabs$ that ensure the path formula $\varphi$ from state $s$. 
By Lemma~\ref{lem:transl_one}(item~2) we have that there exist strategies $\straa'$ and $\strab'$ such that the two plays $\pat' = \plays(s',\straa',\strab')$ 
and $\pat=\plays(s,\straa,\strab)$ satisfy $\omega \sim_{\qual} \omega'$. By inductive hypothesis we have that $s' \models \llangle 1,2 \rrangle (\varphi)$.

\item Consider a path formula $\varphi$. 
If $\pat  \sim_{\qual} \pat'$, then by inductive hypothesis for every sub-formula $\varphi'$ 
of $\varphi$ we have that if $\pat \models \varphi'$ then $\pat'\models \varphi'$.
It follows that if $\pat \models \varphi$ then $\pat'\models \varphi$.

\end{itemize}

\noindent{\em Second implication.}
It remains to prove the second implication $\preccurlyeq_{C}^* \subseteq \preccurlyeq_{C}\subseteq \maxqual$. 
Assume that given states $s$ and $s'$ we have that $(s,s') \not \in \maxqual$, 
then there exists a winning strategy in the corresponding combined-simulation game for the adversary from state $(s,s')$,
i.e., there exists a strategy $\strab^\qual$ such that against all strategies $\straa^\qual$ we have 
$\plays((s,s'),\straa^\qual,\strab^\qual)$ reaches a state labeled $p$.
As memoryless strategies are sufficient for both players in $\game^\qual$~\cite{gradel2002automata}, there also exists a bound $i \in \nat$, 
such that the proponent fails to match the choice of the adversary in at most $i$ turns. 
We sketch the inductive proof that there exists a formula  with $i$ nested operators 
$\llangle1\rrangle \Next$ or $\llangle 1,2 \rrangle \Next$ that is satisfied in $s$ but not in $s'$. 
For $i$ equal to $0$ the states can be distinguished by atomic propositions. 
For the inductive step one can express the simulation turns by a $\llangle1,2 \rrangle (\Next \ldots)$ formula 
and alternating simulation turns by a $\llangle1\rrangle (\Next \ldots )$ formula. 
It follows that $(s,s') \not \in \preccurlyeq_{C}$.
The result follows.
\qed
\end{proof}

\begin{remark}
Lemma~\ref{lem:transl_one} and Theorem~\ref{thm:equiv} also hold for alternating games. 
Note that in most cases the action set is constant and the state space of the games are huge.
Then the combined simulation game construction is quadratic, and solving safety games on them
can be achieved in linear time (on the size of the game) using discrete graph theoretic algorithms~\cite{Immerman81,Beeri}.
\end{remark}

\begin{theorem}
\label{thm:quadratic}
Given two-player games $\game$ and $\game'$, the $\maxqual$, $\preccurlyeq_{C}^*$, and 
$\preccurlyeq_{C}$ relations can be computed in quadratic time using discrete graph theoretic algorithms.
\end{theorem}



\section{MDPs and Qualitative Logics}
\label{sec:mdp}
In this section we consider Markov decisions processes (MDPs) and 
logics to reason qualitatively about them.
We consider MDPs which can be viewed as a variant of two-player games defined in Section~\ref{sec:defn}.
First, we fix some notation: a probability distribution $f$ on a finite set $X$ 
is a function $f:X \to [0,1]$ such that $\sum_{x\in X} f(x)=1$, 
and we denote by  $\distr(X)$ the set of all probability distributions on $X$. 
For $f \in \distr(X)$ we denote by $\supp(f)=\set{x\in X \mid f(x)>0}$ the \emph{support of} $f$.

\subsection{MDPs}
A \emph{Markov decision process} (MDP) is a tuple 
$\game = (\states,(\states_1,\states_P),\act,\av,\trans_1, \trans_P,\lab,s_0)$; where
(i)~$\states$ is a finite set of states with a partition of $\states$ into 
Player-1 states $\states_1$ and probabilistic states $\states_P$;
(ii)~$\act$ is a finite set of actions;
(iii)~$\av: S_1 \to 2^\act \setminus \emptyset$ is an action-available 
function that assigns to every Player-1 state the non-empty set 
$\av(s)$ of actions available in $s$;
(iv)~$\trans_1:\states_1 \times \act \to \states  $ is a deterministic
transition function that given a Player-1 state and an action gives the
next state;
(v)~$\trans_P: \states_P \to \distr(\states)$ is a probabilistic 
transition function that given a probabilistic state gives a probability 
distribution over the successor states (i.e., $\trans_P(s)(s')$ is the 
transition probability from $s$ to $s'$);
(vi)~the function $\lab$ is the proposition labeling function as for two-player games;
and (vii)~$s_0$ is the initial state.
Strategies for Player~1 are defined as for games.
In this work we will consider MDPs with qualitative  properties, and 
hence not consider reward-based MDP models.


\smallskip\noindent\textbf{Interpretations.}
We interpret an MDP in two distinct ways: 
(i)~as a $1\tfrac{1}{2}$-player game and 
(ii)~as an alternating two-player game. 
In the $1\tfrac{1}{2}$-player setting in a state $s \in \states_1$, 
Player~1 chooses an action $a \in \av(s)$ and the MDP moves to a unique successor  $s'$.
In probabilistic states $s_p \in \states_P$ the successor is chosen
according to the probability distribution $\trans_P(s_p)$.
In the alternating two-player interpretation, we regard the probabilistic states as Player-2 states,
 i.e., in a state $s_p \in \states_P$,  Player~2 chooses a successor state $s'$ from the support of the 
probability distribution $\trans_P(s)$. 
Given an MDP $\game$ we denote by $\wh{\game}$ its two-player interpretation,
and $\wh{G}$ is an alternating game.  
The $1\tfrac{1}{2}$-player interpretation is the classical definition of MDPs. 
We will use the two-player interpretation to relate logical characterizations of MDPs
and logical characterization of two-player games with fragments of $\atl^*$.

\smallskip\noindent\textbf{$1\tfrac{1}{2}$-Player Interpretation.}
Once a strategy $\straa \in \straas$ for Player~1 is fixed, the outcome of the MDP is a random walk 
for which the probabilities of \emph{events} are uniquely defined, where an \emph{event}
$\Phi \subseteq \Omega$ is a measurable set of plays~\cite{gradel2002automata}. 
For a state $s\in \states$ and an event $\Phi \subseteq \Omega$, 
we write $\prb^\straa_{s}(\Phi)$ for the probability that a play belongs to $\Phi$ 
if the game starts from the state s and Player~1 follows the strategy~$\straa$.

\smallskip\noindent\textbf{Two-player Interpretation.} The two-player interpretation corresponds to alternating
two-player games introduced in Section~\ref{sec:defn}, 
where  the probabilistic aspect of the MDP is replaced by a second player. 
Formally, given an MDP $\game = (\states,(\states_1,\states_P),\act,\av,\trans_1, \trans_P,\lab,s_0)$ we define an 
alternating two-player game $\wh{\game} = (\wh{\states},\wh{\act},\wh{\av},\wh{\trans},\wh{\lab},\wh{s_0})$ as follows:
(i)~ the states are $\wh{\states} = \states_1 \cup \states_P$; (ii)~the set of actions contains a new action $\bot$ not present
in $\act$, i.e., $\wh{\act} = \act \cup \set{\bot}$; (iii)~the action-available function for states $s \in \states_1$ is defined
as $\wh{\av}(s)  = \av(s)$ and for states $s_p \in \states_P$ as $\wh{\av}(s_p) = \set{\bot}$; 
(iv) for $s \in \states_1$ and  $a$ in $\wh{\av}(s)$ we have $\wh{\trans}(s,a) = \set{\trans_1(s,a)}$,
and for $s_p \in \states_P$ we have $\wh{\trans}(s_p,\bot) = \supp(\trans_p(s_p))$; 
(v)~the labeling function for a Player-1 state $s$ is  $\wh{\lab}(s) = \lab(s) \cup \set{\turn}$ and for a Player-2 state $s'$ coincides with $\lab(s')$; and (vi)~the initial state is the same $\wh{s}_0 = s_0$.
Given an MDP $\game$ we denote by $\wh{\game}$ the two-player interpretation of the MDP.
Note that for all Player-1 states $s \in \states_1$ we have $\vert \wh{\trans}(s) \vert = 1$ and for all Player-2 states 
$s_p \in \states_P$ we have $\vert \av(s_p) \vert = 1$. Therefore for any MDP the corresponding two-player interpretation
is an alternating game.


\smallskip\noindent\textbf{Parallel composition of MDPs.}
An MDP is said to be strictly alternating if the initial state is a Player-1 state and all the successors of Player-1 states are probabilistic states, and vice versa. 
Given two strictly alternating MDPs 
$\game = (\states,(\states_1,\states_P),\act,\av,\trans_1, \trans_P,\lab,s_0)$ and 
$\game' = (\states',(\states'_1,\states'_P),\act,\av',\trans'_1, \trans'_P,\lab',s'_0)$,
the parallel composition  is an MDP $\game \parallel \game' = (\wb{\states},(\wb{\states}_1,\wb{\states}_P) 
,\act, \wb{\av}, \wb{\trans}_1, \wb{\trans}_P, \wb{\lab}, \wb{s}_0)$ defined as follows:
(i)~the states are $\wb{\states} = \wb{\states}_1 \cup \wb{\states}_P$, where $\wb{\states}_1 = \states_1 \times \states'_1$ and
 $\wb{\states}_P = \states_P \times \states'_P$; (ii)~for a state $(s,s') \in \wb{\states}_1$ we have $\wb{\av}((s,s'))  = \av(s) 
 \cap \av'(s')$; (iii)~for a state $(s,s') \in \wb{\states}_1$ and an action $a \in \wb{\av}((s,s'))$ we have $\wb{\trans}_1((s,s'),a)
  = (\trans_1(s,a),\trans'_1(s',a))$; (iv)~for a state $(s_p,s_p') \in \wb{\states}_P$ we have $\wb{\trans}((s_p,s_p'))(t,t')= 
  \trans_P(s_p)(t) \cdot \trans'_P(s_p')(t')$; (v)~ for a state $(s,s') \in \wb{\states}$ we have $\wb{\lab}((s,s')) = \lab(s) \cup \lab'(s')$, 
  and (vi)~the initial state is
$(s_0,s'_0)$.


\begin{example}

\label{ex:mdps}
In Figure~\ref{fig:mdps} we present three MDPs $\game_1, \game_2$, and $\game'$ that we use as running examples. 
We thoroughly describe only MDP~$\game' = (\states,(\states_1,\states_P),\act,\av,\trans_1, \trans_P,\lab,s_0)$.
Player-1 states, depicted as circles,  are $\states_1 = \set{s'_0,s'_2, s_3'}$
and probabilistic states, depicted as rectangles, are $\states_P = \set{s'_1, s'_4}$.
The set of actions is $\act = \set{a,b}$. Action $a$ is available in states $s'_0,s'_2$ and action $b$ is available only in states $s'_0,s_3'$. 
The deterministic transition function is $\trans_1(s'_0, a)=s'_1, \trans_1(s'_0, b)=s'_4, \trans_1(s'_2, a)=s'_4, \trans_1(s'_2, b)=s'_4, \trans_1(s'_3, b)=s'_4$.
The probabilistic transition function $\trans_P$  gives the following probability distributions over possible successor states: 
$\trans_P(s'_1)(s'_2)= \frac{1}{2}, \trans_P(s'_1)(s'_3)= \frac{1}{2}, \trans_P(s'_4)(s'_3)= 1$.
There is a single atomic proposition $p \in \ap$ and the states labeled by $p$ are depicted in gray.
The initial state is $s'_0$.\qed
\end{example}
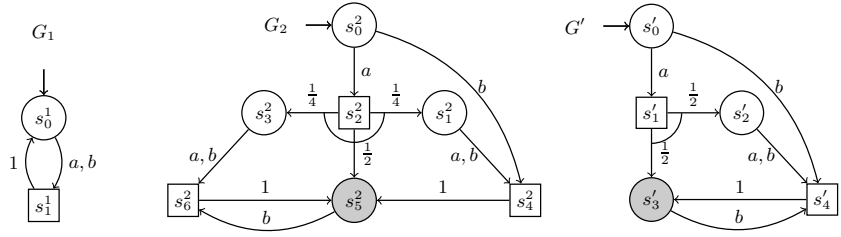
\begin{figure}[htb]
\centering
  \scalebox{0.8}{

\begin{tikzpicture}[auto, node distance=1.5cm,->, semithick,initial text=, align=left]
\tikzstyle{nstate}=[circle, draw]
\tikzstyle{pstate}=[rectangle, draw]
\tikzstyle{initstate}=[nstate,initial]
\tikzstyle{transition}=[->,>=stealth']

\node [](g1){
	
	\begin{tikzpicture}[semithick]
	\node [initstate, initial above] (s10) {$s^1_0$};
	\node [pstate, below=0.8cm of s10] (s11) {$s^1_1$};
	\node [above=0.8cm of s10]   {$\game_1$};

	\draw[->, semithick]
	(s10) edge[bend left] node[right] {$a,b$} (s11)
	(s11) edge[bend left] node[left] {$1$} (s10);
	\end{tikzpicture}
};

\node[right of=g1,xshift=100] (g2){
	\begin{tikzpicture}[semithick]
	\node [initstate, initial left] (s20) at (0,0) {$s^2_0$};
	\node [pstate, below=0.8cm of s20] (s22)  {$s^2_2$};
	\node [nstate, right of=s22] (s21)  {$s^2_1$};
	\node [nstate, fill=gray!40, below=0.8cm of s22] (s25)  {$s^2_5$};
	\node [nstate, left of=s22] (s23)  {$s^2_3$};
	\node [pstate, right=2.2cm of s25] (s24)  {$s^2_4$};
	\node [pstate, left=2.2cm of s25] (s26)  {$s^2_6$};

	\node [left=0.6cm of s20]   {$\game_2$};

        \path[->, semithick]
        (s20) edge node[right] {$a$}    (s22)
        (s22) edge node[above] {$\frac{1}{4}$} (s21)
        (s22) edge node[above] {$\frac{1}{4}$} (s23)
        (s21) edge node[left] {$a,b$}    (s24)
        (s22) edge node[right] {$\frac{1}{2}$}    (s25)
        (s23) edge node[left] {$a,b$}    (s26)
        (s24) edge node[above] {$1$}    (s25)
        (s25) edge[bend left] node[above] {$b$}    (s26)
        (s26) edge node[above] {$1$}    (s25)
        (s20) edge[bend left]node[right] {$b$}	   (s24);

        \draw[-, semithick] (-0.5,-1.45) arc (180:360:0.5);
	\end{tikzpicture}
};

\node[right of=g2,xshift=120] (g'){
	\begin{tikzpicture}[semithick]
	\node [initstate, initial left] (s30) at (0,0) {$s'_0$};
	\node [pstate, below=0.8cm of s30] (s32)  {$s'_1$};
	\node [nstate, right of=s32] (s33)  {$s'_2$};
	\node [nstate, fill=gray!40, below=0.8cm of s32] (s35)  {$s'_3$};
	\node [pstate, right=2.2cm of s35] (s36)  {$s'_4$};

	\node [left=0.6cm of s30]   {$\game'$};

        \path[->, semithick]
        (s30) edge node[right] {$a$}    (s32)
        (s32) edge node[above] {$\frac{1}{2}$} (s33)
        (s32) edge node[right] {$\frac{1}{2}$}    (s35)
        (s33) edge node[left] {$a,b$}    (s36)
        (s35) edge[bend right] node[above] {$b$}    (s36)
        (s36) edge node[above] {$1$}    (s35)
        (s30) edge[bend left]node[right] {$b$}	   (s36);

        \draw[-, semithick] (0.5,-1.45) arc (0:-90:0.5);
\end{tikzpicture}
};

\end{tikzpicture}

 }
  \caption{Examples of MDPs.}
  \label{fig:mdps}
\end{figure}

\subsection{Qualitative Logics for MDPs}
We consider the qualitative fragment of $\pctl^*$~\cite{HJ94,BerkP95,BdA95}  and refer to the logic as \emph{qualitative pCTL$^*$} (denoted as $\qatl^*$) 
as it can express qualitative properties of MDPs.

\smallskip\noindent\textbf{Syntax and semantics.}  
The syntax of the logic is given in positive normal form and is similar to the syntax of $\atl^*$. 
It has the same state and path formulas as $\atl^*$ with the exception of path quantifiers. 
The logic $\qatl^*$ comes with two path quantifiers ($\PQ)$, namely $\almost$ and $\positive$
(instead of $\llangle 1 \rrangle, \llangle 2 \rrangle,\llangle 1,2 \rrangle \text{, and } \llangle \emptyset \rrangle$).
\begin{eqnarray*}
\qatl^* \text{ path quantifiers:} & &  \almost, \positive.
\end{eqnarray*}
The semantics of the logic $\qatl^*$ is the same  for the fragment shared with $\atl^*$, 
therefore we only give semantics for the new  path quantifiers. 
Given a path formula $\varphi$, we denote by $\llbracket \varphi \rrbracket_{\game}$ the set of plays $\pat$ such that $\pat \models \varphi$. 
For a state $s$ and a path formula $\varphi$ we have:
\begin{eqnarray*}
s \models \almost(\varphi) \qquad  & &\text{iff } \exists \straa \in \straas: \prb_s^{\straa}(\llbracket \varphi \rrbracket) = 1 \\
s \models \positive(\varphi) \qquad & &\text{iff } \exists \straa \in \straas: \prb_s^{\straa}(\llbracket \varphi \rrbracket) > 0.
\end{eqnarray*}
As before, we denote by $\qatl$ the fragment of $\qatl^*$ where every temporal operator is immediately preceded by a path quantifier, and for a state formula $\psi$ the set 
$\llbracket \psi \rrbracket_\game$ denotes the set of states in $\game$ that satisfy the formula $\psi$. 

\smallskip\noindent{\bf Logical relation induced by $\qatl$ and $\qatl^*$.}
Given two MDPs $\game$ and $\game'$, the logical relation induced by $\qatl^*$, denoted as 
$\preccurlyeq_{Q}^*$, (resp. by $\qatl$, denoted as $\preccurlyeq_{Q}$), is defined as follows:
$$ \preccurlyeq_{Q}^* = \{(s,s') \in \states \times \states' \mid \forall \psi \in \qatl^*: \text{ if } s \models \psi \text{ then } s' \models \psi\}$$
(resp. $\forall \psi \in \qatl$).

\section{Characterization of Qualitative Simulation for MDPs}
\label{sec:mdplogic}
In this section we establish the equivalence of the $\preccurlyeq_{Q}^{*}$ 
relation on MDPs with the $\preccurlyeq_{C}^*$ relation on the two-player interpretation of
MDPs, i.e., we prove that for all MDPs $\game$ and $\game'$ we have 
$\preccurlyeq_{Q}^{*}(\game,\game') = \preccurlyeq_{C}(\wh{\game},\wh{\game}')$,
where $\wh{\game}$ (resp. $\wh{\game}'$) is the two-player interpretation of
the MDP $\game$ (resp. $\game'$). 
In the first step we show how to translate some of the $\qatl$ formulas into $\catl$ formulas. 
We only need to translate the path quantifiers due to the similarity of path formulas in the logics.

\begin{lemma}
\label{lem:transl}
For all atomic propositions $q,r$ and for all MDPs $\game$, we have:
\setcounter{equation}{0}
\begin{align}
\llbracket \almost (\Next q )\rrbracket_\game &= \llbracket \llangle 1 \rrangle (\Next q) \rrbracket_{\wh{\game}} \\
\llbracket \almost (q \wrel r) \rrbracket_\game &= \llbracket \llangle 1 \rrangle (q \wrel r) \rrbracket_{\wh{\game}} \\
\llbracket \positive (\Next q )\rrbracket_\game &= \llbracket \llangle 1,2 \rrangle (\Next q) \rrbracket_{\wh{\game}} \\
\llbracket \positive (q \until r) \rrbracket_\game &= \llbracket \llangle 1,2 \rrangle (q \until r) \rrbracket_{\wh{\game}}
\end{align}
\end{lemma}
\begin{proof}

\emph{Point 1.} The inclusion $\llbracket \almost (\Next q) \rrbracket \supseteq \llbracket \llangle 1 \rrangle (\Next q) \rrbracket$ follows from the fact
that there exists a strategy for Player~1 such that for all strategies of Player~2 the next state reached satisfies $q$. It follows that the same strategy for Player~1
ensures the formula with probability $1$. For the second inclusion $\llbracket \almost (\Next q) \rrbracket \subseteq \llbracket \llangle 1 \rrangle (\Next q) \rrbracket$
we consider two cases: (i) let $s \in \llbracket \almost (\Next q) \rrbracket$ be a Player-1 state. Then there exists an available action $a$ that leads to
a state that satisfies formula $q$. As $s$ is a Player-1 state, the transition function under $a$ has a unique successor. Therefore, playing the same
action ensures $q$ also in the two-player interpretation. The second case is that $s$ is a probabilistic states. In that case all the successors in the support
of the probabilistic transition function satisfy $q$. Therefore formula $q$ is also satisfied in the two-player interpretation.

%
\emph{Point 2.} As for the previous point the inclusion $\llbracket \almost (q \wrel r) \rrbracket \supseteq \llbracket \llangle 1 \rrangle (q \wrel r) \rrbracket$ follows easily from the definition. For the second inclusion assume towards contradiction that for every strategy $\straa$ for Player~1 there exists a strategy $\strab$ for Player~2 such that the  play $\plays(s,\straa,\strab)$  violates $q \wrel r$. It follows that for every strategy $\straa$ for Player~1 there exists a strategy $\strab$ for Player~2 such that play $\plays(s,\straa,\strab)$ satisfies $\neg r \until \neg q$. This is possible only if there exists a finite path to a $\neg q$ state that uses only $\neg r$ states, and the finite path has a positive probability in the $1 \tfrac{1}{2}$-player interpretation of the MDP. It follows that for every strategy of Player~1 there is a positive probability of violating $q \wrel r$ and the contradiction follows.

\emph{Point 3. and 4.} Point~$3$ follows similarly to Point~1, and Point~$4$ follows the same arguments as in Point~2. 
\qed
\end{proof}

\begin{lemma}
\label{lem:wreltrans}
For all atomic propositions $r$ and for all MDPs we have:
$\llbracket \positive (\Box \: r) \rrbracket = \llbracket \positive(r \until \almost(\Box \: r)) \rrbracket$.
\end{lemma}
\begin{proof}
The result follows from~\cite[Lemma~1]{CDH10} (shown even for a more general class of partially observable MDPs). 
\qed
\end{proof}

\begin{lemma}
\label{lem:transl_w}
For all atomic propositions $q,r$ and for all MDPs, we have:
$\llbracket \positive (q \wrel r) \rrbracket = \llbracket \llangle 1,2 \rrangle (q \until r) \rrbracket \cup \llbracket \llangle 1,2 \rrangle (q \until (\llangle 1 \rrangle( q \wrel \false))) \rrbracket $.
\end{lemma}
\begin{proof}
By definition we have that $\llbracket \positive (q \wrel r) \rrbracket = \llbracket \positive ((q \until r) \vee (\Box q)) \rrbracket$. 
We write the formula as follows:
$\llbracket \positive ((q \until r) \vee (\Box q)) \rrbracket = \llbracket \positive (q \until r) \rrbracket \cup \llbracket \positive (\Box q) \rrbracket$. 
By Lemma~\ref{lem:wreltrans} we have that $ \llbracket \positive (\Box q) \rrbracket = \llbracket \positive(q \until \almost(\Box \: q)) \rrbracket$. 
Note that $\Box \: q \equiv q \wrel \false$. 
All these facts together with the already established translations presented in Lemma~\ref{lem:transl} give us the desired result.
\qed
\end{proof}

To complete the translation of temporal operators it remains to express the $\qatl$ formula $\llbracket \almost (q \until r) \rrbracket$ in terms of $\catl$. We first introduce the $\apre$ function:

\smallskip\noindent\emph{$\apre$.} Given two sets of states $X \subseteq Y \subseteq \states$ we define the predecessor operator $\apre$ as follows:
\begin{eqnarray*}
\apre(Y,X) &=& \set{s \in \states_1 \mid \exists a \in \av(s)\ :\ \trans_1(s,a) \in X }  \cup \\
 & &\set{s_p \in \states_P \mid \supp(\trans_P(s_p)) \subseteq Y \wedge \supp(\trans_P(s_p)) \cap X \neq \emptyset}.
\end{eqnarray*}
As is shown in \cite{dAHK98} we can express the states $\llbracket \almost (q \until r) \rrbracket$ using the following $\mu$-calculus notation, where $\mu$ (resp. $\nu$) denotes the least (resp. greatest) fixpoint:
\begin{equation}
\label{eq:fix}
\llbracket  \almost (q \until r) \rrbracket  = \nu Y.\mu X. (\llbracket r \rrbracket \cup (\llbracket q \rrbracket \cap \apre(Y,X))).
\end{equation}
The fixpoint computation on an MDP with $n$ states can be described as follows: 
$Y_0$ is initialized to all states, and in each iteration $i$ the set $X_{i,0}$ 
is initialized to the empty set; and $X_{i,j+1}$ is obtained from $X_{i,j}$ applying 
the one step operators, and $Y_i$ is set as the fixpoint of iteration $i$.
Formally, for $1\leq i \leq n$ and $0\leq j \leq n-1$ we have
\[
Y_0  = \llbracket  \true \rrbracket; \quad X_{i,0}  =  \llbracket  \false \rrbracket; \quad X_{i,j+1}  =   (\llbracket r \rrbracket \cup (\llbracket q \rrbracket \cap \apre(Y_{i-1},X_{i,j}))); 
\quad Y_i  =  X_{i,n};
\]
and then $Y_n= \llbracket  \almost (q \until r) \rrbracket$.
Next we show that the $\apre$ function can be expressed in $\catl$. 
For $\catl$ formulas $\psi_1,\psi_2$ such that 
$\llbracket \psi_1 \rrbracket \subseteq \llbracket \psi_2 \rrbracket$ 
we define:
$$F_{\apre}(\psi_1,\psi_2)= \llangle1\rrangle (\Next \psi_1) \wedge \llangle1,2 \rrangle (\Next \psi_2) $$

\begin{lemma}
\label{lem:comp}
For $\catl$ state formulas $\psi_1,\psi_2$ such that
$\llbracket \psi_1 \rrbracket \subseteq \llbracket \psi_2 \rrbracket$ 
we have: 
$\llbracket F_{\apre}(\psi_1,\psi_2) \rrbracket = \apre(\llbracket \psi_1 \rrbracket, \llbracket \psi_2 \rrbracket)$.
\end{lemma}
\begin{proof}
We prove the two inclusions. We start with $\apre(\llbracket \psi_1 \rrbracket, \llbracket \psi_2 \rrbracket) \subseteq \llbracket F_{\apre}(\psi_1,\psi_2) \rrbracket$.
Let $s$ be a state in $\apre(\llbracket \psi_1 \rrbracket, \llbracket \psi_2 \rrbracket)$, we consider two cases: (i)~$s \in \states_1$; and (ii)~$s \in \states_P$.
For the case (i)~it follows from the definition of $\apre$ that there exists an action $a \in \av(s)$ such that the unique state $\trans_1(s,a)$ satisfies $\psi_1 \land \psi_2$.
It follows that $s \in \llbracket \llangle1\rrangle (\Next \psi_1) \wedge \llangle1,2 \rrangle (\Next \psi_2)\rrbracket$ and therefore
$s \in \llbracket F_{\apre}(\psi_1,\psi_2) \rrbracket$. In case (ii)~
we have $s \in \states_P$, $\supp(\trans_P(s)) \subseteq \llbracket \psi_1 \rrbracket$, and $\supp(\trans_P(s)) \cap \llbracket \psi_2 \rrbracket \neq \emptyset$.
It follows that $s \in \llbracket \llangle1\rrangle (\Next \psi_1) \wedge \llangle1,2 \rrangle (\Next \psi_2) \rrbracket$ and therefore 
$s \in \llbracket F_{\apre}(\psi_1,\psi_2) \rrbracket$.

We continue with the second inclusion $\llbracket F_{\apre}(\psi_1,\psi_2) \rrbracket \subseteq \apre(\llbracket \psi_1 \rrbracket, \llbracket \psi_2 \rrbracket)$.
Let $s$ be a state in $\llbracket F_{\apre}(\psi_1,\psi_2) \rrbracket$, we again consider two cases: (i)~$s \in \states_1$; and (ii)~$s \in \states_P$.
For case (i)~assume $s \in \llbracket \llangle1\rrangle (\Next \psi_1) \wedge \llangle1,2 \rrangle (\Next \psi_2) \rrbracket$,
it follows that there exists an available action $a \in \av(s)$ such that the state $\trans_1(s,a)$ is in $\llbracket \psi_2 \rrbracket$ and as we 
have $\llbracket \psi_2 \rrbracket \subseteq \llbracket \psi_1 \rrbracket$, we have that there exists an action $a \in \av(s)$ such that $\trans_1(s,a) \in \llbracket \psi_1 \rrbracket \cap \llbracket \psi_2 \rrbracket$. For the second case (ii)~when $s \in \states_P$ we again assume $s \in \llbracket \llangle1\rrangle (\Next \psi_1) \wedge \llangle1,2 \rrangle (\Next \psi_2) \rrbracket$. The first part of the formula ensures that $\trans_P(s) \subseteq \llbracket \psi_1 \rrbracket$ and the second part ensures that $\trans_P(s) \cap \llbracket \psi_2 \rrbracket \neq \emptyset$. The desired result follows.
\qed
\end{proof}

The following lemma shows the first of the two inclusions:
\begin{lemma}
\label{lem:catltoqatl}
For an MDP we have $\preccurlyeq_{C} \: \subseteq \: \preccurlyeq_{Q}$.
\end{lemma}
\begin{proof}
We prove the counterpositive, i.e., we construct a mapping of formulas $f: \qatl \rightarrow \catl$ such that given two states $s,s'$ and a $\qatl$ formula $\psi$ 
we have that if $s \models \psi$ and $s' \not \models \psi$ then the  $\catl$ formula $f(\psi)$ is true in $s$ and not true in $s'$. We proceed by structural induction on the $\qatl$ formula and replace parts that are in scope of a path quantifier by their $\catl$ version. The cases where $\psi$ is an atomic
proposition or a Boolean combination of formulas are straightforward. 
It remains to translate the  formulas $\almost (\Next \varphi_1)$, $\almost (\varphi_1 \wrel \varphi_2)$, and $\almost (\varphi_1 \until \varphi_2)$ for $\qatl$ formulas $\varphi_1, \varphi_2$. The translation of the first two follows directly from Lemma~\ref{lem:transl}, therefore it remains to translate the $\qatl$ formula $\almost(\varphi_1 \until \varphi_2)$. 
We proceed by encoding the fixpoint computation of the $\almost (\varphi_1 \until \varphi_2)$ formula into nested $\catl$ formulas. 
Let $n$ be the number of states of the MDP.
Let $\{\wt{\phi}_i,\ \phi_{i,j} \mid 0 \leq i,j \leq n\}$ be a set of formulas defined by the following clauses:
\begin{eqnarray*}
 & & \wt{\phi}_0= \true; \\
\forall 1 \leq i \leq n: & & \phi_{i,0} = \false  \\
\forall 1 \leq i \leq n. \forall 0 \leq j \leq n-1: & & \phi_{i,j+1} = f(\varphi_2) \vee (f(\varphi_1) \wedge F_{\apre}(\wt{\phi}_{i-1},\phi_{i,j})) \\
\forall 1 \leq i \leq n: & &  \wt{\phi}_{i} =\phi_{i,n};
\end{eqnarray*}
By Lemma~\ref{lem:comp} the set of nested formulas $\phi_{i,j}$ represents the computation of $X_{i,j}$ and $\wt{\phi}_i$ the computation of $Y_i$ 
(for the computation of the fixpoint formula). 
It follows that we have $\llbracket \almost (\varphi_1 \until \varphi_2) \rrbracket = \llbracket \wt{\phi}_{n} \rrbracket$ and concludes the translation. 
The translation for formulas $\positive (\Next \varphi_1)$, $\positive (\varphi_1 \wrel \varphi_2)$, and $\positive (\varphi_1 \until \varphi_2)$
to $\catl$ formulas follows from Lemma~\ref{lem:transl} and Lemma~\ref{lem:transl_w}.
The desired result follows.
\qed
\end{proof}

\begin{lemma}
For an MDP $\game$ we have $\preccurlyeq_{Q} \: \subseteq \: \preccurlyeq_{C}$.
\end{lemma}
\begin{proof}
Given an MDP with $n$ states, it follows from the proof of Theorem~\ref{thm:equiv} for the combined-simulation game that the $n$-step approximation 
$\preccurlyeq_{C}^n$ is exactly the same as $\preccurlyeq_{C}$. We define a sequence $\Psi_0, \Psi_1, \ldots, \Psi_n$ of sets of formulas of $\qatl$ with the property that $s \preccurlyeq_{C}^i t$ iff every formula $\psi \in \Psi_i$ that is true in $s$ is also true in $t$. We denote by $\boolc(\Psi)$ all the formulas that consist of disjunctions and conjunctions of formulas in $\Psi$. We assume that $\boolc(\Psi)$ does not contain repeated elements, therefore from finiteness of $\Psi$ follows finiteness of $\boolc(\Psi)$. We define $\Psi_0 = \boolc(\{q, \neg q \mid q \in \ap\})$, and for all $0 \leq i < n$ we define $\Psi_{i+1} = \boolc(\{\Psi_i \cup \{ \positive (\Next \psi), \almost (\Next \psi) \mid \psi \in \Psi_i\}\})$. The formulas in $\Psi_0, \Psi_1, \ldots, \Psi_n$ provide witnesses that for all $0 \leq i \leq n$ we have that $\preccurlyeq_{Q} \subseteq \preccurlyeq_{C}^i$, in particular we have that $\preccurlyeq_{Q} \subseteq \preccurlyeq_{C}$.
\qed
\end{proof}

\begin{theorem}
\label{thm:catlqatl}
For all MDPs $\game$ and $\game'$ we have $\preccurlyeq_{Q} \: = \: \preccurlyeq_{C}$.
\end{theorem}

\begin{theorem}
\label{thm:extension}
For all MDPs $\game$ and $\game'$ we have $\preccurlyeq_{Q}^{*} \: = \: \preccurlyeq_{Q}$
\end{theorem}
\begin{proof} {\em (Sketch).}
We need to show that if a $\qatl^*$ formula distinguishes two states,
then there is a $\qatl$ formula that also distinguishes them.
The basic idea is similar to the proof of~\cite[Theorem~7.1, assertion~2]{CDFL09}.
We first construct a deterministic parity automata given the formula in 
$\qatl^*$, and the almost-sure or positive solutions for MDPs with parity 
objectives can be encoded as a $\mu$-calculus formula~\cite{CdAH11}.
The translation of $\mu$-calculus formulas to a $\qatl$ formula is done 
as in Lemma~\ref{lem:catltoqatl}.
\qed
\end{proof}

\begin{theorem}
Given an MDP the relation $\preccurlyeq^*_{Q}$ can be computed in 
quadratic time using discrete graph theoretic algorithms.
\end{theorem}
\begin{proof}
Follows directly from Theorems~\ref{thm:quadratic}, \ref{thm:catlqatl}, and \ref{thm:extension}.
\hfill
\qed
\end{proof}



\section{CEGAR for Combined Simulation}
\label{sec:cegar}
In this section we present a CEGAR approach for the computation of
combined simulation.

\subsection{Simulation Abstraction and Alternating-Simulation Abstraction}
\label{sec:abstraction}

\smallskip\noindent\textbf{Abstraction.}
An \emph{abstraction} of a game consists of a partition of the game graph
such that in each partition the atomic proposition labeling match for all 
states. 
Given an abstraction of a game, the abstract game can be defined by  
collapsing states of each partition and redefining the action-available and transition functions.
The redefinition of the action-available and transition functions can 
either increase or decrease the power of the players.
If we increase the power of Player~1 and decrease the power of Player~2,
then the abstract game will be in alternating simulation with the original game,
and if we increase the power of both players, then 
the abstract game will simulate the original game.
We now formally define the partitions, and the two abstractions.

\smallskip\noindent\textbf{Partitions for abstraction.}
A \emph{partition} of a  game $\game=(\states, \act, \av, \trans, \lab, s_0)$ is an equivalence relation 
$\Part=\set{\pi_1,\pi_2,\ldots,\pi_k}$ on $\states$ such that:
(i)~for all $1\leq i \leq k$ we have $\pi_i \subseteq \states$ and for all $s,s'\in \pi_i$ we have $\lab(s)=\lab(s')$ 
(labeling match);
(ii)~$\bigcup_{1\leq i\leq k} \pi_i = \states$ (covers the state space); and 
(iii)~for all $1\leq i,j \leq k$, such that $i\neq j$ we have $\pi_i \cap \pi_j=\emptyset$ (disjoint).
Note that in alternating games Player~1 and Player~2 states are distinguished by proposition $\turn$, so they belong to 
different partitions.

\smallskip\noindent\textbf{Simulation abstraction.}
Given a two-player game $\game=(\states, \act, \av, \trans, \lab, s_0)$ and 
a partition $\Part$ of $\game$, 
we define the \emph{simulation abstraction of} $\game$ as a two-player game 
$\simabs{\game}{\Part}  = (\ov{\states}, \act, \ov{\av}, \ov{\trans}, \ov{\lab}, \ov{s}_0)$, 
where
\begin{itemize}
\item $\ov{\states} = \Part$: the partitions in $\Part$ are the states of the abstract game.
\item For all $\pi_i \in \Part$ we have $\ov{\av}(\pi_i) = \bigcup_{s\in \pi_i} \av(s)$: 
the set of available actions is the union of the actions available to the states in the partition,
and this gives more power to Player~1.
\item For all $\pi_i \in \Part$ and $a \in \ov{\av}(\pi_i)$ we have 
$\ov{\trans}(\pi_i,a) = \set{ \pi_j \mid \exists s \in \pi_i: \ ( a \in \av(s) \wedge \exists s'\in \pi_j: 
\ s'\in \trans(s,a))}$: 
there is a transition from a partition $\pi_i$ given an action $a$ to a partition $\pi_j$ 
if some state $s \in \pi_i$ can make an $a$-transition to some state in $s' \in \pi_j$, and
this gives more power to Player~2.
\item For all $\pi_i \in \Part$ we have $\ov{\lab}(\pi_i)=\lab(s)$ for some $s\in \pi_i$:
the abstract labeling is well-defined, since all states in a partition are labeled
by the same atomic propositions.
\item $\ov{s}_0$ is the partition in $\Part$ that contains state $s_0$.
\end{itemize}

\smallskip\noindent\textbf{Alternating-simulation abstraction.}
Given a two-player game $\game=(\states, \act, \av, \trans, \lab, s_0)$ and 
a partition $\Part$ of $\game$, 
we define the \emph{alternating-simulation abstraction of} $\game$ as a two-player game 
$\altabs{\game}{\Part}  = (\wt{\states}, \act, \wt{\av}, \wt{\trans}, \wt{\lab}, \wt{s}_0)$, where
\begin{itemize}
\item (i)~$\wt{\states} = \Part$;
(ii)~for all $\pi_i \in \Part$ we have $\wt{\av}(\pi_i) = \bigcup_{s\in \pi_i} \av(s)$;
(iii)~for all $\pi_i \in \Part$ we have $\wt{\lab}(\pi_i)=\lab(s)$ for some $s\in \pi_i$;
(iv)~$\wt{s}_0$ is the partition in $\Part$ that contains state $s_0$
(as in the case of simulation abstraction).
\item For all $\pi_i \in \Part$ and $a \in \wt{\av}(\pi_i)$ we have 
$\wt{\trans}(\pi_i,a) = \set{ \pi_j \mid \forall s \in \pi_i: \ (a \in \av(s) \wedge \ \exists s'\in \pi_j: 
\ s'\in \trans(s,a))}$: 
there is a transition from a partition $\pi_i$ given an action $a$ to a partition $\pi_j$ 
if all states $s \in \pi_i$ can make an $a$-transition to some state in $s' \in \pi_j$, and
this gives less power to Player~2.
For technical convenience we assume $\wt{\trans}(\pi_i,a)$ is non-empty.
\end{itemize}
The following proposition states that (alternating-)simulation abstraction of a game $\game$ is in
(alternating-)simulation with $\game$.
\begin{proposition}
\label{proposition:abs}
For all partitions $\Part$ of a two-player game $\game$ we have:
(1)~$\game \altgame  \altabs{\game}{\Part}$; and 
(2)~$\game \simgame \simabs{\game}{\Part}$.
\end{proposition}

\begin{example}
Consider a two-player interpretation of the MDP~$\game_2$ from Figure~\ref{fig:mdps}.
The coarsest partition of $\game_2$ is $\Part=\set{\ppart_0, \ppart_1, \ppart_2}$, where $\ppart_0=\set{s^2_0, s^2_1, s^2_3}, \ppart_1 = \set{s^2_2, s^2_4, s^2_6}, \ppart_2 = \set{s^2_5}$.
The alternating-simulation abstraction and the simulation abstraction of $\Part$ are depicted in Figure~\ref{fig:abs}.
\qed\end{example}
\begin{figure}[htb]
\centering
  \scalebox{0.8}{

\begin{tikzpicture}[auto, node distance=1.5cm,->, semithick,initial text=, align=left]
\tikzstyle{nstate}=[circle, draw]
\tikzstyle{pstate}=[rectangle, draw]
\tikzstyle{initstate}=[nstate,initial]
\tikzstyle{transition}=[->,>=stealth']

\node [](p1){
	
	\begin{tikzpicture}[semithick]
	\node [initstate, initial left] (a0) {$\ppart_0$};
	\node [pstate, right of=a0] (a1) {$\ppart_1$};
	\node [nstate, fill=gray!40, right of=a1] (a2) {$\ppart_2$};

	\node [left=0.4cm of a0]   {$\altabs{\game_2}{\Part}$};

	\draw[->, semithick]
        (a0) edge node[above] {$a,b$} (a1)
        (a1) edge[bend right] node[below] {$\bot$} (a2)
        (a2) edge[bend right] node[above] {$b$} (a1);
	\end{tikzpicture}
};

\node[right of=g1,xshift=150] (p2){
	\begin{tikzpicture}[semithick]
	\node [initstate, initial left] (s0) {$\ppart_0$};
	\node [pstate, right of=a0] (s1) {$\ppart_1$};
	\node [nstate, fill=gray!40, right of=a1] (s2) {$\ppart_2$};

	\node [left=0.4cm of s0]   {$\simabs{\game_2}{\Part}$};
	\node [below=-1mm of s1]   {$\bot$};

	\draw[->, semithick]
        (s0) edge[bend left] node[above] {$a,b$} (s1)
        (s1) edge[bend right] node[above] {} (s2)
        (s1) edge[bend left] node[above] {} (s0)
        (s2) edge[bend right] node[above] {$b$} (s1);
        \draw[-] (0.9,-0.3) arc (220:325:0.7);
	\end{tikzpicture}

};

\end{tikzpicture}

 }
  \caption{Alternating-simulation and simulation abstractions of $\game_2$ from Figure \ref{fig:mdps}.}
  \label{fig:abs}
\end{figure}
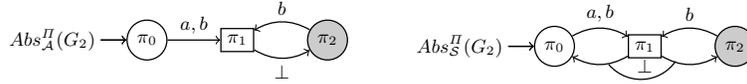

\subsection{Sound Assume-Guarantee Rule}
In this section we  present the sound assume-guarantee rule for
the combined-simulation problem.
To achieve this we first need an extension of the notion of combined-simulation game.

\smallskip\noindent\textbf{Modified combined-simulation games.}
Consider games $\game^\txtalt=(\states, \act, \trans^\txtalt, \av^\txtalt,  \lab, s_0)$, 
$\game^\txtsim=(\states, \act, \trans^\txtsim, \av^\txtsim, \lab, s_0)$ and 
$\game'=(\states', \act, \trans', \av', \lab', s'_0)$.
The \emph{modified simulation game} $\game^{\M} = (\states^\M,\act^\M,\av^\M, \trans^\M,\lab^\M,s^\M_0)$ 
is defined exactly like the combined simulation game given $\game^\txtalt$ and $\game'$, 
with the exception that the step-wise simulation gadget is defined using the transitions of 
$\game^\txtsim$ instead of $\game^\txtalt$.
Formally, we change the transitions as follows:

\begin{itemize}
\item \emph{Checking step-wise simulation conditions.}
Transition (a) is redefined: 
for a state $(s,s',\SI,2)$ we have only one action $\bot$ available for Player~1 and 
we have $\trans^{\M}((s,s',\SI,2),\bot)=\set{(t,s',\SI,1) \mid \exists a \in \av^\txtsim(s): \ t \in \trans^\txtsim(s,a)}$.
\end{itemize}
We write $(\game^\txtalt \otimes \game^\txtsim) \dualgame \game'$ 
if and only if $(s_0, s'_0) \in  \llbracket \llangle 1 \rrangle (\Box \neg p) \rrbracket_{\game^\M}$.

\begin{proposition}
\label{proposition:mcsg}
Let $\game, \game', \game^\txtalt, \game^\txtsim$ be  games such that 
$\game \altgame \game^\txtalt$ and $\game \simgame \game^\txtsim$.
Then $(\game^\txtalt \otimes \game^\txtsim) \dualgame \game'$ 
implies $\game \qualgame \game'$.
\end{proposition}

The key proof idea for the above proposition is as follows: 
if $\game \altgame \game^\txtalt$ and $\game \simgame \game^\txtsim$,
then in the modified combined-simulation game $\game^\M$ 
the adversary (Player~2) is stronger than in the combined-simulation game 
$\game^\qual$. Hence winning in $\game^\M$ for the proponent (Player~1)
implies winning in $\game^\qual$ and gives the desired result of
the proposition.

\smallskip\noindent{\bf Sound assume-guarantee method.}
Given two games $\game_1$ and $\game_2$, checking whether their parallel composition 
$\game_1 \parallel \game_2$ is in combined simulation with a game $\game'$ 
can be done explicitly by constructing the synchronized product.
The composition, however, may be much larger than the components and thus make the method
ineffective in practical cases.
We present an alternative method that proves combined simulation in a compositional manner, 
by abstracting $\game_2$ with some partition $\Part$ and then composing it with $\game_1$.
The sound assume-guarantee rule follows from Proposition~\ref{proposition:abs} and Proposition~\ref{proposition:mcsg}.

\begin{proposition}[Sound assume-guarantee rule]
Given games $\game_1,\game_2, \game'$, and a partition $\Part$ of $\game_2$, let $\altcomp=\game_1 \parallel \altabs{\game_2}{\Part}$
and   $\simcomp = \game_1 \parallel \simabs{\game_2}{\Part}$.
If $(\altcomp\otimes \simcomp) \dualgame \game'$, then $(\game_1 \parallel \game_2) \qualgame \game'$,
i.e., 
\begin{equation}
  \label{eq:ag_rule}
 \infer{(\game_1 \parallel \game_2) \qualgame \game'}{\altcomp = \game_1 \parallel \altabs{\game_2}{\Part}; \ & \ \simcomp = \game_1 \parallel \simabs{\game_2}{\Part}; \ & \ (\altcomp\otimes \simcomp) \dualgame \game'}  
\end{equation}
\end{proposition}

\begin{remark}\label{rem:complete}
Note that for the trivial partition $\Part$, where every equivalence relation is a singleton, 
the modified combined-simulation game coincides with the combined simulation game.
We will use this fact to argue about completeness our CEGAR approach.
\end{remark}

If the partition $\Part$  is coarse, then the abstractions in the assume-guarantee rule can be smaller than 
$\game_2$ and also their composition with $\game_1$.
As a consequence, combined simulation can be proved faster as compared to explicitly computing the composition. 
In Section~\ref{sec:cegar_alg} we describe how to effectively compute the partitions $\Part$ and refine them 
using CEGAR approach.

\subsection{Counterexamples Analysis}
\label{sec:cex}
If the premise $(\altcomp\otimes \simcomp) \dualgame \game'$ of the assume-guarantee rule~\eqref{eq:ag_rule} 
is not satisfied, then the adversary (Player~2) has a memoryless winning strategy in $\game^\M$,
and the memoryless strategy is the \emph{counterexample}.
To use the sound assume-guarantee rule~\eqref{eq:ag_rule} in a CEGAR loop,
we need analysis of counterexamples.

\smallskip\noindent{\em Representation of counterexamples.}
A counterexample is a memoryless winning strategy for Player~2 
in $\game^\M$. 
Note that in $\game^\M$ Player~2 has a reachability objective,
and thus a winning strategy ensures that the target set is always
reached from the starting state, and hence no cycle can be formed 
without reaching the target state once the memoryless winning strategy
is fixed.
Hence we represent counterexamples as directed-acyclic graphs (DAG), 
where the leafs are the target states
and every non-leaf state has a single successor chosen by the strategy of Player~2 
and has all available actions for Player~1.

\smallskip\noindent{\em Abstract, concrete, and spurious counterexamples.}
Given two-player games $\game_1$ and $\game_2$, let 
$\game=(\game_1 \parallel \game_2)$ be the parallel composition.
Given $\game$ and $\game'$, let $\game^\qual$ be the combined-simulation game 
of $\game$ and $\game'$.
The abstract game $\game^\M$  is the modified combined-simulation 
game of $(\altcomp\otimes \simcomp)$ and $\game'$, where 
$\altcomp = \game_1 \parallel \altabs{\game_2}{\Part}$ and $\simcomp = \game_1 \parallel \simabs{\game_2}{\Part}$.
We refer to a counterexample $\strab_{\abs}$ in $\game^\M$ as \emph{abstract}, and 
to a counterexample $\strab_{\con}$ in $\game^\qual$ as  \emph{concrete}.
An abstract counterexample is \emph{feasible} if we can substitute partitions 
in $\altcomp$ and $\simcomp$  with states of $\game_2$ to obtain a concrete counterexample. 
An abstract counterexample is \emph{spurious} if it is not feasible.

\smallskip\noindent{\em Concretization of counterexamples.}
We follow the approach of~\cite{Henzinger03} to check the feasibility of a 
counterexample by finding a \emph{concretization} function $\conc$ from states 
in $\game^\M$ to a set of states in $\game_2$ that  witness a 
concrete strategy from the abstract strategy.
A state in $\game^\M$ has a component which is a partition for $\game_2$, 
and the concretization  constructs a subset of the partition.
Intuitively, for a state $\ov{s}$ of $\game^\M$ in the counterexample DAG, the concretization 
represents the subset of states of $\game_2$ in the partition where a concrete 
winning strategy exists using the strategy represented by the DAG below the state $\ov{s}$.
Informally, the witness concrete strategy is constructed inductively, going bottom-up 
in the DAG as follows:
(i)~the leaves already represents winning states and hence their concretization 
is the entire partition;
(ii)~for non-leaf states in the DAG of the abstract counterexample, 
the concretization represents the set of states of $\game_2$ of the partition 
which lead to a successor state that belongs to the concretization of the 
successor in the DAG.
An abstract counterexample is feasible, if the concretization of the 
root of the DAG contains the initial state of $\game_2$.

\smallskip\noindent{\em Computation of the concretization.}
Given an abstract counterexample $\strab_{\abs}$ and a state $\ov{s}$ in $\game^\M$,
let $\tsucc(\ov{s})$ be the set of all successor of $\ov{s}$ in $\game^\M$ 
given $\strab_{\abs}$ is 
fixed by Player~2.
The formal description of the concretization is given in Figure~\ref{fig:conc},
where the concretization of a state $\ov{s}$ in the abstract counterexample is computed 
from its successors in the DAG.
We use the notation $\av^1$, $\av^2$, and  $\trans^2$ to represent the action-available
functions of $\game_1$ and $\game_2$, and the transition function of $\game_2$,
respectively.

\smallskip\noindent{\em Illustrative examples.}
We present intuitive description of two representative cases of concretization
from Figure~\ref{fig:conc}:
(1)~Consider a state $\ov{s}=((s_1, \pi_2),s',\AL,2)$ where the abstract counterexample chooses 
the successor $\ov{s}'=((s_1, \pi_2),s',\AL,a,1)$ (intuitively this corresponds to choice of action $a$).
The concretization $\conc(\ov{s}) = \set{s \in \pi_2 \mid a \in \av^2(s) \land
s \in \conc(\ov{s}')}$ is the subset of states in $\pi_2$ where the action $a$ is available and 
$s$ also belongs to the concretization of the successor state $\ov{s}'$.
(2)~For a state $\ov{s}=((s_1,\ppart_2),s',\AL,a,a',1)$, the concretization is 
the set of states where action $a$ is not available or 
all successors given action $a$ belong to the concretization of the successors of $\ov{s}$.

\begin{figure}[h]
  \centering
  $\begin{array}{l l l }
    \ov{s}=((s_1,\pi_2),s') &: \ \ \conc(\ov{s})=&
    \begin{cases}
\pi_2 & \text{$\ov{s}$ is a leaf} \\
\conc(\ov{s}') & \text{otherwise, where }\tsucc(\ov{s})=\{\ov{s}'\} \\[0.5ex]      
    \end{cases} \\[2ex]
    \ov{s}=((s_1,\pi_2),s',\SI,2)&: \ \ \conc(\ov{s})=&\set{s\in \pi_2 \mid \exists a\in \av^1(s_1)\cap\av^2(s): \trans^2(s,a)\cap   \conc(\ov{s}') \neq \emptyset}\\
    				 &                    &\text{ where }\tsucc(\ov{s})=\set{\ov{s}'} \\[2ex]
    \ov{s}=((s_1,\pi_2),s',\SI,1)&: \ \ \conc(\ov{s})=& \displaystyle \bigcap_{\ov{s}'\in \tsucc(\ov{s})} \conc(\ov{s}') \\[2ex]
    \ov{s}=((s_1,\pi_2),s',\AL,2)&: \ \ \conc(\ov{s})=&\{s\in \pi_2 \mid a\in \av^2(s) \land s\in \conc(\ov{s}'),  \} \text{ where }\\[0.5ex]
                        & & \tsucc(\ov{s})=\{\ov{s}'\}\text{ and }\ov{s}'=((s_1,\pi_2),s',\AL,2,a)\\[2ex]
    \ov{s}=((s_1,\pi_2),s',\AL,a,1)&: \ \ \conc(\ov{s})=& \displaystyle \bigcap_{\ov{s}'\in\tsucc(\ov{s})} \conc(\ov{s}') \\[2ex]
    \ov{s}=((s_1,\pi_2),s',\AL,a,a',2)&: \ \ \conc(\ov{s})=& \conc(\ov{s}'), \text{ where }\tsucc(\ov{s})=\{\ov{s}'\} \\[2ex]
     \ov{s}=((s_1,\pi_2),s',\AL,a,a',1)&: \ \ \conc(\ov{s})=& \displaystyle \{s\in \pi_2 \;|\; a \not\in \av^2(s) \vee
\trans^2(s,a)\subseteq  \bigcup_{\ov{s}'\in\tsucc(\ov{s})} \conc(\ov{s}') \}
  \end{array}$
  \caption{Concretization function; $\ov{s}$ is a state in an abstract counterexample.}
  \label{fig:conc}
\end{figure}

\begin{figure}[htb]
  \centering
  \scalebox{1}{\begin{tikzpicture}[auto, node distance=3cm,->, semithick,initial text=, align=left, font=\tiny]
\tikzstyle{state}=[rectangle, draw, inner sep=0.5mm]
\tikzstyle{state2}=[ inner sep=1mm]
\tikzstyle{slab}=[font=\small]
\tikzstyle{initstate}=[state,initial]
\tikzstyle{transition}=[->,>=stealth']
\def\pact{\bot}
\node [initstate, initial left] (t0) {$((s^1_0,\ppart_0),s'_0)$};
\node (t1) [state, below=0.4cm of t0]  {$((s^1_0,\ppart_0),s'_0,\SI,2)$};
\node [state, below=0.4cm of t1] (t2) {$((s^1_1,\ppart_1),s'_0,\SI,1)$};
\node [state, below right=0.5cm of t2] (t3) {$((s^1_1,\ppart_1),s'_1)$};
\node [state2, below=0.4cm of t3] (t4) {$\ldots$};

\node [state, below left=0.5cm of t2] (q3) {$((s^1_1,\ppart_1),s'_4)$};
\node [state2, below=0.4cm of q3] (q4) {$\ldots$};

\node [state2, font=\normalsize, right=0.9cm of t0] (c0) {$\emptyset$};
\node [state2, font=\normalsize, right=0.9cm of t1] (c1) {$\emptyset$};
\node [state2, font=\normalsize, right=0.9cm of t2] (c2) {$\emptyset$};
\node [state2, font=\normalsize, right=0.9cm of t3] (c3) {$\{s^2_4, s^2_6\}$};
\node [state2, font=\normalsize, left=0.9cm of q3] (cq3) {$\{s^2_2\}$};

\path
        (t0) edge node[right]  {$\bot$} (t1)
        (t1) edge node[right] {$\bot$}  (t2)
        (t2) edge node[below] {$s'_1$}  (t3)
        (t3) edge node[right]  {$\bot$} (t4)
        (t2) edge node[below] {$s'_4$}  (q3)
        (q3) edge node[right]  {$\bot$} (q4);

\path[dashed, font=\scriptsize]
	(t0) edge node[above] {$\conc$} (c0)
	(t1) edge node[above] {$\conc$} (c1)
	(t2) edge node[above] {$\conc$} (c2)
	(t3) edge node[above] {$\conc$} (c3)
	(q3) edge node[above] {$\conc$} (cq3);

\end{tikzpicture} }
  \caption{Abstract counterexample to the modified combined-simulation game of $(\altcomp\otimes \simcomp)$ and $\game'$, 
where $\altcomp = \wh{\game_1} \parallel \altabs{\wh{\game_2}}{\Part}$ and $\simcomp = \wh{\game_1} \parallel \simabs{\wh{\game_2}}{\Part}$. }
  \label{fig:cex}
\end{figure}

\begin{example}
\label{ex:spurious}
Consider MDPs $\game_1, \game_2, \game'$ in Figure~\ref{fig:mdps} interpreted as games and the abstract games $\altabs{\wh{\game_2}}{\Part}$, $\simabs{\wh{\game_2}}{\Part}$ in Figure~\ref{fig:abs}.
Let $\altcomp = \wh{\game_1} \parallel \altabs{\wh{\game_2}}{\Part}$ and $\simcomp = \wh{\game_1} \parallel \simabs{\wh{\game_2}}{\Part}$.
Figure \ref{fig:cex} shows part of an abstract counterexample to the modified combined-simulation game of $(\altcomp\otimes \simcomp)$ and $\game'$. 
In this counterexample the adversary first plays in the simulation gadget and the proponent responds by moving to a state $((s^1_1,\ppart_1),s'_1)$ or a state $((s^1_1,\ppart_1),s'_4)$ (their successors are not depicted in Figure~\ref{fig:cex}). 
From the state $((s^1_1,\ppart_1),s'_1)$ the adversary has a winning strategy by playing in the alternating-simulation gadget, and from $((s^1_1,\ppart_1),s'_4)$ by playing in the simulation gadget.
The dashed shows assign the concretization of states in the abstract counterexample.
The counterexample is spurious, since the initial state of $\game_2$ does not belong to the concretization of the initial state of the counterexample.\qed
\end{example}

\label{sec:ref}
\begin{algorithm}[h]
  \caption{Assume-guarantee CEGAR for $\qualgame$.}
  \label{alg:rg}
\begin{algorithmic}
\Require Two-player games $\game_1, \game_2, \game'$.
\Ensure \textbf{yes} if $\game_1 \parallel \game_2 \qualgame \game'$, otherwise \textbf{no}
\State $\Part \gets $ coarsest partitioning of $\game_2$
\Loop
	\State $\altcomp \gets \game_1 \parallel \altabs{\game_2}{\Part};$ $ \quad \simcomp \gets \game_1 \parallel \simabs{\game_2}{\Part}$
  \State $\game^\M \gets $ modified combined simulation game of $(\altcomp\otimes \simcomp)$ and $\game'$
	\If {Player~1 wins in $\game^\M$}
\quad  \Return \textbf{yes}
	\Else
        	\State $\Cex \gets $abstract counterexample in $\game^\M$
                \If {Feasible($\Cex$)}
\quad \Return \textbf{no}
                \Else
\quad  $\Part \gets$ Refine($\Cex$, $\Part$)
                \EndIf
	\EndIf
\EndLoop
\end{algorithmic}
\end{algorithm}

\subsection{CEGAR}
\label{sec:cegar_alg}
The counterexample analysis presented in the previous section allows us to 
automatically refine abstractions using the CEGAR 
paradigm \cite{Clarke00}.
The code of the CEGAR algorithm for the assume-guarantee combined simulation is shown in Algorithm~\ref{alg:rg}.
The algorithm takes $\game_1, \game_2, \game'$
as arguments and answers whether $(\game_1 \parallel \game_2) \qualgame \game'$ holds.
Initially, the algorithms computes the coarsest partition $\Part$ of $\game_2$. 
Then, it executes the CEGAR loop:
in every iteration the algorithm constructs $\altcomp$ (resp. $\simcomp$) as the parallel composition of $\game_1$ and 
the alternating-simulation abstraction (resp. simulation abstraction) of $\game_2$.
Let $\game^\M$ be the modified combined-simulation game of $(\altcomp \otimes \simcomp)$ and $\game'$.
If Player~1 has a winning strategy in $\game^\M$ then the algorithm  returns YES; 
otherwise it finds  an abstract counterexample $\Cex$ in $\game^\M$.
In case the counterexample is feasible, then  it corresponds to a concrete counterexample, and 
the algorithm returns NO.
If $\Cex$ is spurious, the algorithm calls a refinement procedure that uses the concretization of $\Cex$
to return a partition  $\Part'$ finer than partition $\Part$.
Our technique can be extended to handle multiple components in a similar way as presented 
in~\cite[Section~5]{Komuravelli12}.

\smallskip\noindent\textbf{Refinement procedure.}
Given a partition $\Part$ and a spurious counterexample $\Cex$ together with its concretization function 
$\conc$ we describe how to compute the refined partition $\Part'$.
Consider a partition $\ppart \in \Part$ and let $\ov{\states}_\ppart = 
\set{\ov{s}_1,\ov{s}_2, \ldots, \ov{s}_m}$ denote the states of the abstract counterexample
$\Cex$ that contain $\ppart$ as its component. 
Every state $\ov{s}_i$ splits $\ppart$ into at most two sets 
$\conc(\ov{s}_i)$ and $\ppart \setminus \conc(\ov{s}_i)$, and 
let this partition be denoted as $T_i$.
We define a partition $\mathcal{P}_{\ppart}$ as the largest equivalence relation on 
$\ppart$ that is finer than any of the equivalence relation $T_i$ for all $1 \leq i \leq m$.
Formally, $\mathcal{P}_{\ppart} = \set{\ov{\pi}_1,\ov{\pi}_2, \ldots,\ov{\pi}_k}$ is a 
partition of $\ppart$ such that for all $1 \leq j \leq k$ and $1\leq i \leq m$ we have 
$\ov{\pi}_j \subseteq \conc(\ov{s}_i)$ or 
$\ov{\pi}_j \subseteq \ppart \setminus \conc(\ov{s}_i)$. 
The new partition $\Part'$ is then defined as the union over $\mathcal{P}_{\ppart}$ 
for all $\ppart \in \Part$.

\begin{example}
We continue with our running example. 
In Example~\ref{ex:spurious} we showed that the abstractions
 of $\wh{\game_2}$ by the coarsest partition $\Part$ lead to a spurious counterexample depicted 
 in Figure~\ref{fig:cex}. 
Consider the partition $\ppart_1 = \set{s^2_2,s^2_4,s^2_6}$.
There are three states in the counterexample
that have $\ppart_1$ as its component and the concretization function assigns to them three subsets of states: 
$\emptyset, \set{s^2_2}, \set{s^2_4,s^2_6}$. After the refinement partition $\ppart_1$ is split 
into two partitions $\ppart'_1 = \set{s^2_2}$ and $\ppart''_1 = \set{s^2_4,s^2_6}$.
\qed
\end{example}

\begin{proposition}
Given a partition $\Part$ and a spurious counterexample $\Cex$, the partition $\Part'$ obtained 
as refinement of $\Part$  is finer than $\Part$.
\end{proposition}

\smallskip\noindent{\bf Sound and completeness of our CEGAR approach.}
Since we consider finite games, the refinement procedure only executes for finitely many 
steps. 
In every iteration of the CEGAR algorithm, either the algorithm returns a correct answer 
(by soundness), or a finer partition is obtained.
Thus either we end up with a correct answer, or the trivial partition, and hence 
by Remark~\ref{rem:complete} the completeness of our approach follows.
Thus our CEGAR approach is both sound and complete.


\newcommand{\cs}{\mathsf{CS}}
\newcommand{\sn}{\mathsf{SN}}
\newcommand{\mer}{\mathsf{MER}}
\newcommand{\rdp}{\mathsf{RDP}}
\newcommand{\lep}{\mathsf{LE}}
\newcommand{\ec}{\mathsf{EC}}
\newcommand{\petp}{\mathsf{PETP}}
\newcommand{\petg}{\mathsf{PETG}}
\newcommand{\viruss}{\mathsf{VIR1}}
\newcommand{\viruse}{\mathsf{VIR2}}
\newcommand{\agc}{\mathsf{AGCS}}
\newcommand{\agas}{\mathsf{AGAS}}
\newcommand{\ags}{\mathsf{AGSS}}
\newcommand{\monc}{\mathsf{MONCS}}
\newcommand{\mons}{\mathsf{MONSS}}
\newcommand{\monas}{\mathsf{MONAS}}
\newcommand{\memout}{\mbox{MO}}
\newcommand{\timeout}{\mbox{TO}}
\newcommand{\error}{Error}
\newcommand{\mb}{MB}
\vspace{-1em}
\section{Experimental Results}
\label{sec:impl}
\vspace{-0.5em}
We implemented our CEGAR approach for combined simulation in Java, 
and experimented with our tool on a number of MDPs and two-player games examples.
We use PRISM~\cite{prism} model checker to specify the examples and generate 
input files for our tool.

\smallskip\noindent\emph{Observable actions.}
To be compatible with the existing benchmarks (e.g. \cite{Komuravelli12}) 
in our tool actions are observable instead of atomic propositions. 
Our algorithms are easily adapted to this setting.
We also allow the user to specify silent actions for components,
which are not required to be matched by the specification $\game'$. 

\smallskip\noindent\emph{Improved (modified) combined-simulation game.}
We leverage the fact that MDPs are interpreted as alternating games to simplify the 
(modified) combined-simulation game.
When comparing two Player-1 states, the last two steps in the alternating-simulation gadget can be omitted, 
since the players have unique successors given the actions chosen in the first two steps.
Similarly, for  two probabilistic states, the first two steps in the alternating-simulation 
gadget can be skipped.
We check the (modified) combined-simulation games using the standard attractor 
algorithm to solve games with safety 
(as well as reachability) objectives~\cite{RajeevTomBook,Zie98}.

\smallskip\noindent\emph{Improved partition refinement procedure.}
In the implementation we adopt the approach of~\cite{Henzinger03} for refinement.
Given a state $\ov{s}$ of the abstract counterexample with partition $\ppart$ as its component, 
the equivalence relation may split the set $\ppart \setminus \conc(\ov{s})$ into multiple 
equivalence classes. 
Intuitively, this ensures that similar-shaped spurious counterexamples do not reappear in the 
following iterations. 
This approach is more efficient than the naive one, and also implemented in our tool.

\smallskip\noindent\textbf{MDP examples.}
We used our tool on all the MDP examples from~\cite{Komuravelli12}:

\begin{compactitem}
\item \emph{$\cs_1$ and $\cs_n$} model a Client-Server protocol with mutual exclusion with probabilistic failures in one or all of the $n$ clients, 
respectively.

\item \emph{$\mer$} is an arbiter module of NASAs software for Mars Exploration Rovers which grants shared resources for several users. 

\item \emph{$\sn$} models a network of sensors that communicate via a bounded buffer with probabilistic behavior in the components.
\end{compactitem}
In addition, we also considered two other classical MDP examples:
\begin{compactitem}
\item \emph{$\lep$} is based on a PRISM case study~\cite{prism} that models 
the \emph{Leader election protocol}~\cite{IR90}, where $n$ agents on a ring randomly pick a 
number from a pool of $K$ numbers. 
The agent with the highest number becomes the leader. 
In case there are multiple agents with the same highest number the election proceed to the next round. 
The specification requires that two leaders cannot be elected at the same time.
The MDP is parametrized by the number of agents and the size of the pool. 

\item \emph{$\petp$} is based on a Peterson's algorithm~\cite{P81} for mutual exclusion of $n$ threads, 
where the execution order is controlled by a randomized scheduler.
The specification requires that two threads cannot access the critical section at the same time.
We extend Peterson's algorithm by giving the threads a non-deterministic choice to restart before 
entering the critical section.
The restart operation succeeds with probability $\frac{1}{2}$ and  with probability $\frac{1}{2}$ 
the thread enters the critical section.
\end{compactitem}

\noindent{\em Details of experimental results.}
Table~\ref{tab:mdp_results} shows the results for MDP examples we obtained using our 
assume-guarantee algorithm and the monolithic approach (where the composition is computed explicitly). 
We also compared our results with the tool presented in \cite{Komuravelli12} that implements 
both assume-guarantee and monolithic approaches for \emph{strong simulation}~\cite{SL95}.
All the results were obtained on a Ubuntu-13.04 64-bit machine running on an Intel Core i5-2540M CPU 
of 2.60GHz. 
We imposed a 4.3GB upper bound on Java heap memory and one hour time limit.
For $\mer(6)$ and $\petp(5)$ PRISM cannot parse the input file (probably it runs out of memory).

\smallskip\noindent{\em Summary of results.}
For all examples, other than the Client-Server protocol, the assume-guarantee method scales better
than the monolithic reasoning; and in all examples our qualitative analysis scales better than 
the strong simulation approach.
Qualitative analysis through combined simulation relies on graph-theoretic 
algorithms (attractor computation), while checking strong simulation requires 
calls to an SMT solver.


\begin{table}[ht]
\centering
 \resizebox{\linewidth}{!}{
\begin{tabular}{cccc|cccc|cccc|cc|cc}
\hline
\multicolumn{4}{c}{} & \multicolumn{4}{|c}{$\agc$} & \multicolumn{4}{c|}{$\ags$} & \multicolumn{2}{c}{$\monc$} & \multicolumn{2}{c}{$\mons$} \\
Ex. & $\: \vert \game_1 \vert$ &  $\vert \game_2 \vert $ & $\vert \game' \vert$ & $Time$ & $Mem$ & $I$ & $\vert \Part \vert$ & $Time$ & $Mem$ & $I$ & $\vert \Part \vert$ & $Time$ & $Mem$ & $Time$ & $Mem$\\
\hline
$\cs_1(5)$ & 36 & 405 & 16 &  		1.13s & 112\mb & 49 & 85 & 		6.11s & 213\mb & 32 & 33 & 	\textbf{0.04s} & \textbf{34\mb} & 		0.18s & 95\mb \\
$\cs_1(6)$ & 49 & 1215 & 19 &  		2.52s & 220\mb & 65 & 123 & 		11.41s & 243\mb & 40 & 41 & 	\textbf{0.04s} & \textbf{51\mb} & 		0.31s & 99\mb \\
$\cs_1(7)$ & 64 & 3645 & 22 &  		5.41s  & 408\mb & 84 & 156 & 		31.16s & 867\mb & 56 & 57 & 	\textbf{0.05s} & \textbf{82\mb} & 		0.77s & 113\mb\\
\hline
$\cs_n(3)$ & 125 & 16 & 54 &  		0.65s & 102\mb &9 & 24 & 		33.43s & 258\mb & 11 & 12 & 	\textbf{0.09s} & \textbf{35\mb} & 		11.29s & 115\mb\\
$\cs_n(4)$ & 625 & 25 & 189  &  	6.22s & 495\mb & 15 & 42 & 		\timeout & - & - & - & 		\textbf{0.4s} & \textbf{106\mb} & 		1349.6s & 577\mb\\
$\cs_n(5)$ & 3k & 36 & 648 &  		117.06s & 2818\mb & 24 & 60 & 		\timeout & - & - & - & 		\textbf{2.56s} & \textbf{345\mb} & 		\timeout & -\\
\hline
$\mer(3)$ & 278 & 1728 & 11 &  		\textbf{1.42s} & \textbf{143\mb} & 8 & 14 & 	2.74s & 189\mb & 6 & 7 & 		1.96s & 228\mb & 		128.1s & 548\mb\\
$\mer(4)$ & 465 & 21k & 14 &  		\textbf{4.63s} & \textbf{464\mb} & 13 & 22 & 	10.81s & 870\mb & 10 & 11 & 		11.02s & 1204\mb & 		\timeout & -\\
$\mer(5)$ & 700 & 250k & 17  &  	\textbf{29.23s} & \textbf{1603\mb} & 20 & 32 & 	67s & 2879\mb & 15 & 16 & 		- & \memout& 					\memout & -\\
\hline
$\sn(1)$ & 43 & 32 & 18 &  		0.13s & 38\mb & 3 & 6 & 	0.28s & 88\mb & 2 & 3 & 		\textbf{0.04s} & \textbf{29\mb} & 		3.51s & 135\mb\\
$\sn(2)$ & 796 & 32 & 54 &  		0.9s & 117\mb & 3 & 6 & 	66.09s & 258\mb & 2 & 3 & 		\textbf{0.38s} & \textbf{103\mb} & 		3580.83s & 1022\mb\\
$\sn(3)$ & 7k & 32 & 162 &  		4.99s & \textbf{408\mb} & 3 & 6 & 	\timeout & - & - & - & 		4.99s & 612\mb & 			\timeout & -\\
$\sn(4)$ & 52k & 32 & 486 &  		\textbf{34.09s} & \textbf{2448\mb} & 3 & 6 & 	\timeout & - & - & - & 	 44.47s & 3409\mb & 					\timeout & -\\
\hline
$\lep(3,4)$ & 2 & 652 & 256 &  		\textbf{0.24s} & \textbf{70\mb} & 6 & 14 & 	1.63s & 223\mb & 6 & 7 & 		0.38s & 103\mb & 		\timeout & -\\
$\lep(3,5)$ & 2 & 1280 & 500 &  	\textbf{0.31s} & \textbf{87\mb} & 6 & 14 & 	\error & - & - & - & 			1.77s & 253\mb & 		\error & -\\
$\lep(4,4)$ & 3 & 3160 & 1280 &  	\textbf{0.61s} & \textbf{106\mb} & 6 & 16 & 	\timeout & - & - & - &	 		9.34s & 1067\mb & 		\timeout & -\\
$\lep(5,5)$ & 4 & 18k & 12k &  		\textbf{3.37s} & \textbf{364\mb} & 6 & 18 & 	\timeout & - & - & - & 			- & \memout & 			\timeout & -\\
$\lep(6,4)$ & 5 & 27k & 20k &  		\textbf{6.37s} & \textbf{743\mb} & 6 & 20 & 	\timeout & - & - & - & 			- & \memout & 			\timeout & -\\
$\lep(6,5)$ & 5 & 107k & 78k &  	\textbf{23.72s} & \textbf{2192\mb} & 6 & 20 & 	\timeout & - & - & - & 			- & \memout & 			\timeout & -\\
\hline
$\petp(2)$    & 68 & 3 & 3 &		0.04s & 31\mb & 0 & 2	& 			0.04s & 87\mb & 0 & 1&			0.04s & \textbf{30\mb} 			& 0.04s & 90\mb \\
$\petp(3)$    & 4 & 1730 & 4 &		\textbf{0.19s} & \textbf{65\mb} & 6 & 8	& 	0.29s & 153\mb & 3 & 4&			0.24s & 72\mb 	& 1.07s & 170\mb \\
$\petp(4)$    & 5 & 54k & 5 &		\textbf{1.58s}  & \textbf{325\mb} & 8 & 10& 	3.12s & 727\mb & 4 & 5&			7.04s & 960\mb 	& 31.52s & 1741\mb \\
\hline
\end{tabular}
 }
\caption{Results for MDPs examples: 
$\agc$ stands for our assume-guarantee combined simulation;
$\ags$ stands for assume-guarantee with strong simulation;
$\monc$ stands for our monolithic combined simulation;
and $\mons$ stands for monolithic strong simulation.
The number $I$ denotes the number of CEGAR iterations and $\vert \Part \vert$ the 
size of the abstraction in the last CEGAR iteration.
\timeout\ and \memout\ stand for a time-out and memory-out, respectively, and 
\error\ means that an error occurred during execution.
The memory consumption is obtained using the Unix \texttt{time} command.
}
\label{tab:mdp_results}
\vspace{-3.5em}
\end{table}

\begin{table}[ht]
\centering
 \resizebox{\linewidth}{!}{
\begin{tabular}{cccc|cccc|cc|cccc|cc}
\hline
\multicolumn{4}{c}{} & \multicolumn{4}{|c}{$\agc$} & \multicolumn{2}{c}{$\monc$}  & \multicolumn{4}{|c}{$\agas$} & \multicolumn{2}{c}{$\monas$}    \\
Ex. & $\: \vert \game_1 \vert$ &  $\vert \game_2 \vert $ & $\vert \game' \vert$ & $Time$ & $Mem$ & $I$ & $\vert \Part \vert$ & $Time$ & $Mem$ & $Time$ & $Mem$ & $I$ & $\vert \Part \vert$ &   $Time$ & $Mem$\\
\hline
$\ec(32,6,16)$ & 71k & 193 & 129 &	3.55s & 446\mb & 1 & 7&		\textbf{1.15s} & 281\mb  &	2.34s & 391\mb & 0 & 2 &	1.03s & \textbf{251\mb}	 \\
$\ec(64,7,16)$ & 549k & 385 & 257 &	70.5s & 3704\mb & 1 & 131&	9.07s & 1725\mb  &	16.79s & 1812\mb & 0 & 2 &	\textbf{4.83s} & \textbf{1467\mb}	 \\
$\ec(64,8,16)$ & 1.1m & 769 & 513 &	- & \memout & - & -&		- & \memout  &		\textbf{52.63s} & \textbf{3619\mb} & 0 & 2 &	- & \memout	 \\
$\ec(64,8,32)$ & 1.1m & 1025 & 513 &	- & \memout & - & -&		- & \memout  &		\textbf{54.08s} & \textbf{3665\mb} & 0 & 2 &	- & \memout	 \\
\hline
$\petg(2)$ & 3 & 52 & 3 &		0.08s & 35\mb & 4 & 6&		0.03s & 30\mb  &	0.07s & 35\mb & 4 & 6 &		0.03s & \textbf{29\mb} \\
$\petg(3)$ & 4 & 1514 & 4 &		\textbf{0.2s} & 63\mb & 6 & 8&		0.25s & 74\mb  &	0.22s & \textbf{62\mb} & 6 & 8 &	 	0.21s & 64\mb	 \\
$\petg(4)$ & 5 & 49k & 5 &		1.75s & 316\mb & 8 & 10&	8.16s & 1080\mb  &	\textbf{1.6s} & \textbf{311\mb} & 8 & 10&	6.94s & 939\mb	 \\
\hline
$\viruss(12)$ & 14 & 4097 & 1 &		0.91s & 159\mb & 15 & 30 &	1.69s & 255\mb  & 	\textbf{0.35s} & \textbf{114\mb} & 2 & 4 &	 	1.53s & 215\mb\\
$\viruss(13)$ & 15 & 8193 & 1 &		1.47s & 197\mb & 16 & 32 &	4.36s & 601\mb  & 	\textbf{0.6s} & \textbf{178\mb} & 2 & 4 &	 	2.8s & 402\mb\\
$\viruss(14)$ & 16 & 16k & 1 &		3.09s & 326\mb & 17 & 34 &	8.22s & 992\mb  &	\textbf{0.75s}& \textbf{241\mb} & 2 & 4 &	 	6.49s & 816\mb\\
$\viruss(15)$ & 17 & 32k & 1 &		4.47s & 643\mb & 18 & 36 &	15.13s & 2047\mb  & 	\textbf{1.05s} & \textbf{490\mb} & 2 & 4 &	 	9.67s & 1361\mb\\
$\viruss(16)$ & 18 & 65k & 1 &		8.65s & 1015\mb & 19 & 38 &	41.28s & 3785\mb  & 	\textbf{1.37s} & \textbf{839\mb} & 2 & 4 &	 	23.71s & 2591\mb\\
$\viruss(17)$ & 19 & 131k & 1 &		18.68s & 1803\mb & 20 & 40 &	- & \memout  &		\textbf{2.12s} & \textbf{1653\mb} & 2 & 4 &		62.24s & 4309\mb \\
$\viruss(18)$ & 20 & 262k & 1 &		38.68s & 3079\mb & 21 & 42 &	- & \memout  & 		\textbf{3.35s} & \textbf{2878\mb} & 2 & 4 &	 	- & \memout\\
\hline
$\viruse(12)$ & 13 & 4096 & 1 &		1.02s & 151\mb & 19 & 34 &	0.81 & 154\mb  & 	0.68s & \textbf{122\mb} & 9 & 14 &	 	\textbf{0.57s} & 133\mb\\
$\viruse(13)$ & 14 & 8192 & 1 &		1.48s & 190\mb & 20 & 36 &	1.13s & 216\mb  & 	1.01s & \textbf{183\mb} & 9 & 14 &	1.01s & 208\mb\\
$\viruse(14)$ & 15 & 16k & 1 &		2.9s & 315\mb & 21 & 38 &	2.33s & 389\mb  &	\textbf{1.94s} & \textbf{311\mb} & 9 & 14 &	2.09s & 388\mb	 \\
$\viruse(15)$ & 16 & 32k & 1 &		5s & 631\mb & 22 & 40 &		6.29s & 964\mb  & 	\textbf{2.12s} & \textbf{489\mb} & 9 & 14 &	4.69s & 757\mb\\
$\viruse(16)$ & 17 & 65k & 1 &		9.82s & 949\mb & 23 & 42 &	7.55s & 1468\mb  & 	\textbf{3.96s} & \textbf{897\mb} & 9 & 14 &	6.09s & 1315\mb\\
$\viruse(17)$ & 18 & 131k & 1 &		23.33s & 1815\mb & 24 & 44 &	23.54s & 3012\mb  &	\textbf{8.16s} & \textbf{1676\mb} & 9 & 14 &	15.36s & 2542\mb\\
$\viruse(18)$ & 19 & 262k & 1 &		45.89s & 3049\mb & 25 & 46 &	55.28s & 4288\mb  & 	\textbf{20.3s} & \textbf{2875\mb} & 9 & 14 & 	28.79s & 3755\mb\\
\hline
\end{tabular}
 }
\caption{Results for two-player games examples.}
\label{tab:games_results}
\vspace{-2.5em}
\end{table}

\smallskip\noindent\textbf{Two-player games examples.}
We also experimented with our tool on several examples of games,
where one of the players controls the choices of the system and the other player represents 
the environment.

\begin{compactitem}
\item \emph{$\ec$} is based on~\cite{CCHR12} and models an error-correcting device that sends and receives data blocks over a communication channel.
Notation $\ec(n,k,d)$ means that a data block consists of $n$ bits and it encodes $k$ bits of data; 
value $d$ is the minimum Hamming distance between two distinct blocks. 
In the first component Player~2 chooses a message to be sent over the channel and is allowed to flip some bits in the block during the transmission. 
The second component restricts the number of bits that Player~2 can flip.
The specification requires that every message is correctly decoded.

\item \emph{$\petg$} is the Peterson's algorithm \cite{P81} example for MDPs,
with the following differences:
(a)~the system may choose to restart instead of entering the critical section;
(b)~instead of a randomized scheduler we consider an adversarial scheduler.  
As before, the specification requires mutual exclusion.

\item \emph{$\viruss$} models a virus that attacks a computer system with $n$ nodes
(based on case study from PRISM~\cite{prism}). 
Player~1 represents the virus and is trying to infect as many nodes of the network as possible. 
Player~2 represents the system and may recover an infected node to an uninfected state.
The specification requires that the virus has a strategy to avoid being completely erased, 
i.e., maintain at least one infected node in the network.
\emph{$\viruse$} is a modified version of $\viruss$ with two special critical nodes in the network. 
Whenever both of the nodes are infected, the virus can overtake the system.
The specification is as for $\viruss$, i.e., the virus can play such that at least one node in the 
network remains infected, but it additionally requires that even if the system cooperates with the virus, 
the system is designed in a way that the special nodes will never be infected at the same time.
\end{compactitem}

The results for two-player game examples are shown in Table~\ref{tab:games_results}.
Along with $\agc$ and $\monc$ for assume-guarantee and monolithic combined simulation,
we also consider $\agas$ and $\monas$ for assume-guarantee and monolithic alternating 
simulation, as for properties in $\oneatl$ it suffices to consider only alternating 
simulation. 
For all the examples, the assume-guarantee algorithms scale better than the monolithic ones. 
Combined simulation is finer than alternating simulation and therefore combined simulation 
may require more CEGAR iterations.




\smallskip\noindent{\bf Concluding remarks.}
In this work we considered compositional analysis of MDPs for qualitative 
properties and presented a CEGAR approach.
Our algorithms are discrete graph theoretic algorithms.
An interesting direction of future work would be to consider symbolic 
approaches to the problem.

\smallskip\noindent{\bf Acknowledgements.}
We thank Anvesh Komuravelli for sharing his implementation with us.


\clearpage
\appendix
\section{Technical appendix}

We start with an example that shows that also for alternating games combined simulation is finer that the 
intersection of simulation and alternating-simulation relation.
\begin{figure}[htb]
  \centering
  \scalebox{0.8}{\begin{tikzpicture}[auto, node distance=2cm,->, semithick,initial text=, align=left]
\tikzstyle{nstate}=[circle, draw]
\tikzstyle{pstate}=[rectangle, draw]
\tikzstyle{transition}=[->,>=stealth']
\node [nstate, initial, initial left] (s0) at (0,0) {$s_0$};
\node [pstate] (s1) at (1.5,0.8) {$s_1$};
\node [pstate] (s3) at (1.5,-0.8) {$s_3$};
\node [nstate, fill=gray!40] (s2) at (3,0) {$s_2$};

\node [nstate, initial, initial left] (t0) at (5,0) {$t_0$};
\node [pstate] (t1) at (6.5,0.8) {$t_1$};
\node [pstate] (t3) at (6.5,-0.8) {$t_3$};
\node [nstate, fill=gray!40] (t2) at (8.5,0.8) {$t_2$};
\node [nstate] (t4) at (8.5,-0.8) {$t_4$};

\node  (labs) at (-1,0) {$\game$};
\node  (sr) at (4,0) {$\game'$};

\node [below=-1mm of t1] {$\bot$};
\node [below=-1mm of s1] {$\bot$};

\path
	(s0) edge[bend left] node[above] {$a_1$} (s1)
	(s0) edge[bend right] node[above] {$a_2$} (s3)
	(s1) edge[] node[above] {} (s2)
	(s1) edge[] node[above] {} (s0)
	(s2) edge[bend right] node[above] {$a_3$} (s1)
	(s3) edge[bend right] node[above] {$\bot$} (s0);

\path
	(t0) edge[bend left] node[above] {$a_1$} (t1)
	(t0) edge[] node[below] {$a_2$} (t3)
	(t1) edge[] node[above] {} (t0)
	(t1) edge[] node[below] {} (t2)
	(t2) edge[bend right] node[above] {$a_3$} (t1)
	(t3) edge[] node[above] {$\bot$} (t4)
	(t4) edge[bend left] node[below] {$a_3$} (t3);

\draw [-] (5.9,0.48) arc (220:350:0.7cm);

\draw [-] (0.9,0.48) arc (220:325:0.7cm);



\end{tikzpicture} }
  \caption{Games $\game, \game'$ such that $\game \simgame \game'$ and $\game \altgame \game'$, but $\game \not\qualgame \game'$.}
  \label{fig:combined_alt}
\end{figure}
\begin{example}
Figure \ref{fig:combined_alt} shows two alternating games $\game, \game'$, 
where the circular states belong to Player~1  and the rectangular states belong to Player~2,
white nodes are labeled by proposition $p$ and gray nodes by proposition $q$.
The largest simulation and alternating-simulation relations between $\game$ and $\game'$ are:
$\simul_{\max}=\set{(s_0, t_0),(s_1, t_1),(s_2, t_2), (s_3, t_1)}, \alt_{\max}=\set{(s_0, t_0),(s_0, t_4),(s_2, t_2), (s_3, t_3), (s_1, t_3), (s_1, t_1)}$. 
Formula $\llangle 1 \rrangle (\Box(p \land \llangle 1,2 \rrangle (\true \, \until \,  q)))$ is satisfied in  state $s_0$, but not in state $t_0$, hence $(s_0, t_0)\not\in \qual_{\max}$.\qed
\end{example}

We now present detailed proofs of Lemma~\ref{lem:transl_one} and Theorem~\ref{thm:equiv} in the context of alternating games.

\begin{lemma}
\label{lem:transl_one_alt}
Given two alternating games $\game$ and $\game'$, let $\maxqual$ be the combined simulation.
For all $(s,s') \in \maxqual$ the following assertions hold: 
\begin{enumerate}
\item For all Player~1 strategies $\straa$ in $\game$, 
there exists a Player~1  strategy $\straa'$ in $\game'$ such that for every play 
$\pat' \in \plays(s',\straa')$ there exists a play $\pat \in \plays(s,\straa)$ 
such that $\pat \sim_\qual \pat'$.
\item For all pairs of strategies $\straa$ and $\strab$ in $\game$, 
there exists a pair of strategies $\straa'$ and $\strab'$ in $\game'$ 
such that $\plays(s,\straa,\strab) \sim_\qual \plays(s',\straa',\strab')$,
\end{enumerate}
\end{lemma}
\begin{proof}

\smallskip\noindent\emph{Assertion 1.}
As the states of Player~1 and Player~2 are distinguished by the $\turn$ atomic proposition, it follows 
from the fact that $(s,s') \in \maxqual$, that either (i)~$s \in \states_1$ and $s' \in \states'_1$ or (ii)~$s \in \states_2$ and $s' \in \states'_2$.

For the first case (i)~we
consider a winning strategy $\straa^{\qual}$ in $\game^\qual$ such that for all $(s,s') 
\in \maxqual$ and against all strategies $\strab^\qual$ we have 
$\plays((s,s'),\straa^\qual,\strab^\qual) \in \llbracket \Box(\neg p) \rrbracket_{\game^\qual}$.
Given the Player~1 strategy $\straa$ in $\game$ we construct $\straa'$ in 
$\game'$ using the strategy $\straa^\qual$.
Let $h$ be an arbitrary history in $\game^\qual$ that visits only states of type $(\states \times \states')$ that are in $\maxqual$ and ends in $(s,s')$.
Consider a history $w \cdot s$ in $\game$ and $w'\cdot s'$ in $\game'$.
Let $\straa(w \cdot s) = a$, we define $\straa'(w' \cdot s')$ as action $a' = \straa^{\qual}(h \cdot ((s,s'),\AL,2)
 \cdot ((s,s'),\AL,a,2))$, i.e., action $a'$ corresponds
to the choice of the proponents winning strategy $\straa^\qual$  in response to the adversarial choice 
of checking step-wise alternating-simulation followed by action $a$ in $\game$. As both $s$ and $s'$ 
are Player-1 states we have that $\vert \trans(s,a) \vert=1$ and $\vert \trans'(s',a') \vert=1$. Let $(t,t')$ 
be the unique state reached in $2$ steps from $((s,s'),\AL,a,a',2)$ in $\game^\qual$. Assume towards contradiction
 that $\lab^\qual((t,t')) = \{ p \}$, then there exists a strategy for adversary that reaches a loosing state
  while the proponent plays a winning strategy $\straa^\qual$  and the contradiction follows. For the second case
   (ii)~we have that states $s$ and $s'$ belong to Player~2, and there is a single action available for $\straa'$. 

\smallskip\noindent\emph{Assertion 2}
The proof is similar to the first assertion, and instead of using the step-wise 
alternating-simulation gadget for strategy construction (of the first item) 
we use the step-wise simulation gadget from $\game^\qual$ to construct the 
strategy pairs.
\end{proof}

\begin{theorem}
For all alternating games $\game$ and $\game'$ we have
$\qual_{\max} = \preccurlyeq_{C}^*= \preccurlyeq_{C}$.
\end{theorem}

\begin{proof}

\noindent{\em First implication.}
We first prove the implication $\maxqual \subseteq \preccurlyeq_{C}^*$. 
We will show the following assertions:
\begin{itemize}
\item For all states $s$ and $s'$ such that $(s,s') \in \maxqual$, we have that every $\catl^*$ state formula satisfied in $s$ is also satisfied in $s'$.
\item For all plays $\pat$ and $\pat'$ such that $\pat \sim_{\qual} \pat'$, 
we have that every $\catl^*$ path formula satisfied in $\pat$ is also satisfied in $\pat'$.
\end{itemize}
We will prove the theorem by induction on the structure of the formulas.
The interesting cases for the induction step are formulas $\llangle 1 \rrangle (\varphi)$ and $\llangle 1,2 \rrangle (\varphi)$, 
where $\varphi$ are path formulas.

\begin{itemize}
\item Assume $s \models \llangle 1 \rrangle (\varphi)$ and $(s,s') \in \maxqual$. 
It follows that there exists a strategy $\straa \in \straas$ that ensures the path formula $\varphi$
from state $s$ against any strategy $\strab \in \strabs$. 
We want to show that $s' \models \llangle 1 \rrangle (\varphi)$. 
By Lemma~\ref{lem:transl_one_alt}(item~1) we have that there exists a strategy $\straa'$ 
for Player~1 from $s'$ such that for every play $\pat' \in \plays(s',\straa')$ 
there exists a play $\pat \in \plays(s,\straa)$ such that $\pat \sim_{\qual} \pat'$. 
By inductive hypothesis we have that $s' \models \llangle 1 \rrangle (\varphi)$.

\item Assume $s \models \llangle 1,2 \rrangle (\varphi)$ and $(s,s') \in \maxqual$. 
It follows that there exist strategies $\straa \in \straas, \strab \in \strabs$ that ensure the path formula $\varphi$ from state $s$. 
By Lemma~\ref{lem:transl_one_alt}(item~2) we have that there exist strategies $\straa'$ and $\strab'$ such that the two plays $\pat' = \plays(s',\straa',\strab')$ 
and $\pat=\plays(s,\straa,\strab)$ satisfy $\omega \sim_{\qual} \omega'$. By inductive hypothesis we have that $s' \models \llangle 1,2 \rrangle (\varphi)$.

\item Consider a path formula $\varphi$. 
If $\pat  \sim_{\qual} \pat'$, then by inductive hypothesis for every sub-formula $\varphi'$ 
of $\varphi$ we have that if $\pat \models \varphi'$ then $\pat'\models \varphi'$.
It follows that if $\pat \models \varphi$ then $\pat'\models \varphi$.

\end{itemize}

\noindent{\em Second implication.}
It remains to prove the second implication $\preccurlyeq_{C}^* \subseteq \preccurlyeq_{C}\subseteq \maxqual$. We prove that from
the assumption that $(s,s') \not \in \maxqual$ we can construct a $\catl$ formula $\varphi$ such that $s \models \varphi$ and $s' \not \models \varphi$. We 
refer to the formula $\varphi$ as a distinguishing formula.
Assume that given states $s$ and $s'$ we have that $(s,s') \not \in \maxqual$, 
then there exists a winning strategy in the corresponding combined-simulation game for the adversary from state $(s,s')$,
i.e., there exists a strategy $\strab^\qual$ such that against all strategies $\straa^\qual$ we have 
$\plays((s,s'),\straa^\qual,\strab^\qual)$ reaches a state labeled by $p$.
As memoryless strategies are sufficient for both players in $\game^\qual$~\cite{gradel2002automata}, there also exists a bound $i \in \nat$, 
such that the proponent fails to match the choice of the adversary in at most $i$ turns. We construct the $\catl$ formula $\varphi$ inductively:
\begin{itemize}
\item[Base case:] Assume $(s,s') \not \in \maxqual$ and let $0$ be the number of turns the adversary needs to play in order to win.
It follows that $(s,s')$ is a winning state for the adversary, i.e., $\lab^{\qual}((s,s')) = \set{p}$. It follows that $\lab(s) \neq \lab'(s')$.
There are two options: (i)~there exists an atomic proposition $q \in \ap$ that is true in $s$ and not true in $s'$ and distinguishes the two states, or 
(ii)~there exists an atomic proposition $q \in \ap$ that is not true in $s$ and true in $s'$, in that case the formula $\neg q$ distinguishes the two states.
\item[Induction step:] Assume $(s,s') \not \in \maxqual$ and let $n+1$ be the number of turns the adversary needs to play in order to win. As the states of Player~1 and Player~2 are distinguished by the $\turn$ atomic proposition, it follows that either (i)~$s \in \states_1$ and $s' \in \states'_1$ or (ii)~$s \in \states_2$ 
and $s' \in \states'_2$. Otherwise the adversary could win in $0$ turns from $(s,s')$.

We first consider case (i), i.e., $(s,s') \in \states_1 \times \states'_1$. The adversary can choose whether to verify (1)~step-wise alternating-simulation ($\AL$) or
(2)~step-wise simulation ($\SI$).
 After that he chooses an action $a$ to be played according the adversarial strategy $\strab^\qual$ in state $(s,s')$, such that no matter what the proponent plays,
  the adversary will win in $n$ turns. We consider two cases: (1)~the adversary checks for step-wise alternating-simulation relation ($\AL$), or (2)~the adversary
  checks for step-wise simulation relation ($\SI$). For case (1)~we have that there exists an action $a$ for the adversary such that for all actions 
  $a'$ of the proponent the adversary can win in $n$ turns from the unique successor $(t,t')$ of $(s,s')$ given $\AL$ and $a$ was played by the adversary
  and $a'$ by the proponent. From the induction hypothesis there exists a $\catl$ formula $\varphi_n$ such that $t \models \varphi_n$ and $t' \not \models \varphi_n$. 
  We define the formula $\varphi_{n+1}$ that distinguishes states $s$ and $s'$ as $\llangle 1 \rrangle (\Next \varphi_n)$. For case (2), where the adversary plays $\SI$ the proof is exactly the same, as step-wise simulation turn from Player~1 states coincides with step-wise alternating-simulation turn.

Next we first consider case (ii), i.e., $(s,s') \in \states_2 \times \states'_2$. The adversary can choose
 whether to verify (1)~step-wise alternating-simulation ($\AL$) or(2)~step-wise simulation ($\SI$). 
We start with first case (1): there is a unique action $a$ available to the adversary from state $((s,s'),\AL,2)$ and similarly a unique action $a'$ for
the proponent from $((s,s'),a,\AL,1)$. The adversary chooses an action $t'$ from the $((s,s'),a,a',\AL,2)$ according to the winning strategy and the proponent chooses
some action $t_i$ from a set of available successor $(t_1,t_2, \ldots, t_m)$. As the adversary follows a winning strategy $\strab^\qual$ we have that it wins from
all states $(t_i,t')$ for $1 \leq i \leq m$ in at most $n$ turns. From the induction hypothesis there exist $\catl$ formulas $\varphi^i_n$ such that $t_i \models \varphi^i_n$ and $t' \not \models \varphi^i_n$.
We define the formula $\varphi_{n+1}$ that distinguishes states $s$ and $s'$ as $\llangle 1 \rrangle (\Next (\bigvee\limits_{1 \leq i \leq m} \varphi^i_n)$.
For case (2) where the adversary verifies the step-wise simulation step, the proof is analogous. The formula that distinguishes states $s$ and $s'$ is $\llangle 1,2 \rrangle ((\Next \bigvee\limits_{1 \leq i \leq m} \varphi^i_n))$.
\end{itemize}

The desired result follows. \qed

\end{proof}

\end{document}